\documentclass[journal]{IEEEtran}

\usepackage[T1]{fontenc}% optional T1 font encoding
\usepackage{amsmath}
\usepackage{amsthm}

\usepackage{amssymb,mathrsfs}
\usepackage{bm}
\usepackage{algorithm,algorithmic}
\usepackage{bbm}
\usepackage{graphicx}
\usepackage{caption}
\usepackage{subfigure}
\usepackage{calc}
\usepackage{epstopdf}
\newtheorem{thm}{Theorem}
\newtheorem{lm}{Lemma}
\newtheorem{remark}{Remark}

\newtheorem{df}{Definition}
\newtheorem{cond}{Assumption}
\newtheorem{example}{Example}
\usepackage{verbatim}
\usepackage{array,cite}

\usepackage{xcolor}
\usepackage{lipsum}

\usepackage{algorithmic}

\begin{document}

%\hyphenation{IEEE International Symposium on Information Theory 2017}

\title{Statistical Robust Chinese Remainder Theorem for Multiple Numbers}

\author{Hanshen Xiao, Nan Du, Zhikang T. Wang and Guoqiang Xiao

\thanks{Hanshen Xiao is with CSAIL and the EECS Department, MIT, Cambridge, USA. E-mail: hsxiao@mit.edu.}
\thanks{Nan Du is with Department of Statistics, Harvard University, Cambridge, USA. E-mail: nandu@mit.edu}
\thanks{Zhikang T. Wang is with the Department of Physics, University of Tokyo, 7-3-1 Hongo, Bunkyo-ku, Tokyo, Japan. E-mail: wang@cat.phys.s.u-tokyo.ac.jp }
\thanks{Guoqiang Xiao is with the College of Computer and Information Science, Southwest University, Chongqing, China. E-mail: gqxiao@swu.edu.cn}}

% <-this % stops a space
%}

%\hyphenation{IEEE Transactions on Signal Processing}

%\markboth{IEEE Transactions on Signal Processing}%

\let\lc\langle
\let\rc\rangle

\date{}
\markboth{IEEE Transactions on Signal Processing}%
{Shell \MakeLowercase{\textit{et al.}}: Bare Demo of IEEEtran.cls for IEEE Journals}
\date{}

\maketitle

\begin{abstract}
Generalized Chinese Remainder Theorem (CRT) is a well-known approach to solve ambiguity resolution related problems. In this paper, we study the robust CRT reconstruction for multiple numbers from a view of statistics. To the best of our knowledge, it is the first rigorous analysis on the underlying statistical model of CRT-based multiple parameter estimation. To address the problem, two novel approaches are established. One is to directly calculate a conditional maximum a posteriori probability (MAP) estimation of the residue clustering, and the other is based on a generalized wrapped Gaussian mixture model to iteratively search for MAP of both estimands and clustering. Residue error correcting codes are introduced to improve the robustness further. Experimental results show that the statistical schemes achieve much stronger robustness compared to state-of-the-art deterministic schemes, especially in heavy-noise scenarios.
\end{abstract}


\begin{section}{Introduction}

\noindent Pioneered by Xia's remarkable works \cite{xia1999}, \cite{dynamic-1}, there is a rich line of works to advance the understanding of number theory based sparse sensing. Due to physical limitation, estimations with integer ambiguity solution are frequently encountered in many practical scenarios. Such cases include frequency determination of undersampled waveforms \cite{li2009fast,TSP2019,closed,multistage,towards}, phase unwrapping \cite{xia2007phase,li2008fast} etc., which can be modeled by solving some Diophantine equations. Therefore, algebraic approaches can be applied as an alternative way to those classic problems. Representatively, Chinese Remainder Theorem (CRT) based reconstruction and co-prime or nested based sampling/arrays \cite{coprime2010, coprime2014, nested2010} are two successful examples. In particular, due to the nature of the distributed representation of a number with its residues, CRT based reconstructions have been further used in applications such as localization estimation in wireless networks \cite{mle2,wang2011robust}, detection of moving targets using multi-frequency antenna array Synthetic Aperture Radar (SAR) \cite{SAR1,SAR2,SAR3,SAR4}, which can be even executed distributively. Generally speaking, the underlying problem can be described as follows.

 \textbf{Problem of interests:} Consider a set of $N$ numbers, $\bm{Y}_{[1:N]}=\{Y_1, Y_2, ... ,Y_N\}$, and $L$ fixed moduli $\bm{m}_{[1:L]}=\{m_1, m_2, ... , m_L\}$, which are all assumed to be integers temporarily. For each $m_l$, one may observe an unordered set $\bm{R}_{[1:N],l} = \{R_{1l}, R_{2l}, ... ,R_{Nl}\}$, where $R_{il} = \langle Y_i + \Delta_{il} \rangle_{m_l}$, i.e., the residue of $Y_i$ modulo $m_l$ is perturbed by a noise $\Delta_{il}$. Here, $\langle A \rangle_B$ denotes the residue of $A$ modulo $B$ and $\Delta_{il}$ are assumed to be independently and identically distributed (i.i.d.) Gaussian noises for each $i$. $\bm{R}_{[1:N],l}$ are assumed to be unordered, which implies that the correspondences between the elements $R_{il}$ in $\bm{R}_{[1:N],l}$ and $\bm{Y}_{[1:N]}$ are unknown. The ultimate goal is to robustly reconstruct $\bm{Y}_{[1:N]}$ using $\bm{R}_{[1:N],l}$, $l=1,2,...,L$.

The model described above captures a large class of problems. Suppose a sinusoidal signal of multiple frequencies $\bm{Y}_{[1:N]}$ is undersampled with multiple rates, $\bm{m}_{[1:L]}$, and Fourier transform is conducted on the $L$ sample sequences to get frequency spectrums, respectively. From the locations of peaks in the spectrums, which may be perturbed with noise, one can estimate the residues of frequency $Y_i$ modulo $m_l$. However, the correspondence relationship between the peaks, represented by $R_{il}$, in the spectrum and $\bm{Y}_{[1:N]}$ are unknown due to the modulo operation, where the disambiguation is shown to be a nontrivial problem \cite{error-1,RCRT-two,TSP2018}. The distance estimation via multi-frequency phase measurement can also be the case as above, where $Y_i$ stands for the distance while $m_l$ represents the carrier wavelength \cite{mle2,wang2011robust}.

 \textbf{Prior Art}: On the whole, the underlying challenges are twofold: the correspondence ambiguity and perturbation. Though there have been limited interactions cutting across both, each subproblem has been well studied separately.

For the single number case, i.e., $N=1$, the number reconstruction is usually called robust CRT (RCRT), where residues are perturbed with errors. Error control of Hamming-weighted errors in residue codes dated back to 1960s \cite{rns1963} and the first polynomial time decoding scheme was proposed in \cite{goldreich1999}. Nonetheless, in our case, small errors may occur across all observations, $\{R_{il}\}$, and we are more interested in errors bounded with infinity norm. To this end, the first closed-form RCRT for errors with a bounded magnitude was proposed in \cite{closed}. Generalized versions can be found in \cite{multistage}, \cite{xu2014solving}. By deploying non-co-prime moduli, where $m_l = \Gamma M_l$ such that $\{M_l, l=1,...,L\}$ are pairwise co-prime and $\Gamma$ can be a real number, it is demonstrated that when $|\Delta_{il}| < \frac{\Gamma}{4}$ for each $i$ and $l$, the reconstruction error is upper bounded by $\frac{\Gamma}{4}$ as well \cite{closed}. The proof has been shown in \cite{sp2017}, where such bound and modulus selections are optimal.

For multiple numbers with errorless residues, the reconstruction is previously termed generalized CRT (GCRT). The main focus is the largest dynamic range $D$ for $\bm{Y}_{[1:N]}$ such that for arbitrary $\bm{Y}_{[1:N]} \in [0, D)^N$, they can be uniquely determined from their unordered residues modulo $\bm{m}_{[1:L]}$.   The first generic lower bound of $D$ was given in \cite{dynamic-1}, and further sharpened by \cite{dynmaic-sharpen}. In particular, when the residues of $\bm{Y}_{[1:N]}$ modulo $m_l$ are distinct for $l \in \{1,2,...,L\}$, polynomial time GCRT exists \cite{dynamic-2}. So far, the closed form of $D$ is only known when $N=2$ \cite{dynmaic-two}.

To tackle correspondence ambiguity and perturbation simultaneously, generalized robust CRT(GRCRT) has been developed in \cite{RCRT-two} for $N=2$ and later generalized to arbitrary $N$ \cite{TSP2018}. Analogously, under the same setup, it is found when $\Gamma \prod_{l=1}^L M_l = O(D^N)$ and $\max_{i,l} |\Delta_{il}|<\frac{\Gamma}{4N}$, $\bm{Y}_{[1:N]}$ can be uniquely and robustly reconstructed with deviation bounded by $\frac{\Gamma}{4N}$ \cite{TSP2018}.

 \textbf{Motivation:}  CRT suggests that, given a set of moduli $\bm{m}_{[1:L]}$, there is a bijection between a non-negative number $X$ and its residues modulo $\bm{m}_{[1:L]}$ when $X$ is less than the least common multiple (lcm) of $\bm{m}_{[1:L]}$. Said another way, the lcm of $\bm{m}_{[1:L]}$ is the maximal utilization of moduli for number reconstruction. However, when $N \geq 2$, the severe limitation of prior works mainly arises from the large redundancy required to disambiguate residues and tolerate perturbation. When the moduli are in a form $m_l = \Gamma M_l$, the robust deterministic reconstruction relies on the assumption that all $|\Delta_{il}|$ should be bounded by $\frac{\Gamma}{4N}$. On the other hand, either in \cite{TSP2018} or \cite{dynmaic-sharpen}, it trades off the utilization rate of moduli by shrinking the dynamic range $D$ to $O(lcm(m_{[1:L]})^{1/N})$, to uniquely determine the correspondences between $\bm{R}_{[1:N],l}$ and $\bm{Y}_{[1:N]}$, where the number of moduli $L$ is proportional to the number of estimands $N$. Consequently, as $N$ increases, which incurs a larger $L$, $\Gamma$ has to be sharply enlarged as well to meet the error bound in considerable probability. Apparently, it is a paradox that existing schemes behave even worse with more samples obtained.

Clearly, the unknown correspondences between residues and reconstructed numbers are the essential bottleneck. When the correspondences are known, the reconstruction of $N$ numbers is simplified to apply RCRT for a single number $N$ times. Indeed, as explained in \cite{closed}, RCRT matches the maximal moduli utilization rate. Therefore, it is natural to ask whether GRCRT can also achieve such maximal moduli utilization rate. In this paper, resorting to statistics, we answer this question affirmatively.

 \textbf{Contribution and Organization:} To the best of our knowledge, the underlying statistical model of GRCRT has not been systematically studied. The most closely related work is the maximum likelihood estimation (MLE) based RCRT explored in \cite{mle1}, \cite{mle2} for a single number. In this paper,

\begin{enumerate}
	\item We show GRCRT can be described by a generalized wrapped Gaussian Mixture Model (GMM) with extra information on sampling. A systematic statistical analysis is presented.
	\item Any successful estimation depends both on reliable statistical inference and a computationally efficient implementation. We propose two efficient algorithms to address the problem. In Algorithm 1, we first derive the maximum a posteriori probability (MAP) of residue clustering under Assumption 1 in a semi-closed form and the problem is thus reduced to $N$ independent conventional RCRT. In addition, inspired by $K$-means clustering, we further propose Algorithm 2 as an iterative scheme to approximate the MAP of both reconstructed $Y_i$ and residue clustering in general.
	\item We show that the tradeoff amongst the three primary parameters, $N$, $L$ and $\Gamma$, can be further improved by incorporating error correcting codes against outliers. Thorough simulation results show that the statistical schemes significantly improve the performance compared with deterministic methods, especially for the high-noise case. For the extremely low-noise case, the deterministic methods may outperform the proposed methods, which is consistent with the theoretical analysis.
\end{enumerate}

The rest of the paper is organized as follows. In Part II, the background and methology of the proposed schemes are presented. In part III, the MAP of residue clustering is analyzed and we prove the optimal solution can be expressed in a semi-closed form under Assumption \ref{cond}. Part IV develops a framework of generalized GMM and an expectation maximization (EM) based scheme is proposed to approximate the MAP of both clustering and estimands. In Part V, the simulation results of the performance comparison and parameter tradeoff are presented. Residue codes are further introduced to tolerate clustering errors. We conclude and provide future prospects in Part VI.

\end{section}

\section{Background \& Methodology}
\noindent First of all, we specify the problem formally. Given moduli $\bm{m}_{[1:L]} = \{m_l = \Gamma M_l | l=1,...,L\}$, where $M_l$ are pairwise co-prime, there are $N$ numbers, denoted by $\bm{Y}_{[1:N]}= \{Y_i | i=1,...,N\}$, to be reconstructed. For the $l^{th}$ sampler with a modulus $m_l$ as sampling rate, an unordered sample set, $\{R_{il},i=1,2,...,N\}$, is obtained, where $R_{il}$ is the residue of $Y_i$ interfered with noise $\Delta_{il}$ modulo $m_l$, i.e., $R_{il} = \langle Y_i + \Delta_{il} \rangle_{m_l}$. Here, $\Delta_{il}$, $i=1,2,...,N$, is i.i.d. Gaussian noise following $\mathcal{N}(0,\sigma_l)$. We define $\bm{R}_{[1:N],l} = (R_{1l}, R_{2l}, ... ,R_{Nl})$ for short. Furthermore, we assume $\bm{Y}_{[1:N]}$ are independently and uniformly distributed in $[0,D)$, where $D=\Gamma \prod_{l=1}^L M_l$, i.e., the lcm of $\bm{m}_{[1:L]}$, termed as the dynamic range. In the following, let $\hat{\bm{Y}}_{[1:N]}$ denote the estimations of $\bm{Y}_{[1:N]}$.

\begin{remark}
	It is noted that the definition of erroneous residues in our paper is a generalization of Wang and Xia's prior works \cite{closed}, which assumes $R_{il} = \langle Y_i \rangle_{m_l} + \Delta_{il}$. Such definition ignores the cases when $ \langle Y_i \rangle_{m_l} + \Delta_{il}<0$ or $ \langle Y_i \rangle_{m_l} + \Delta_{il}\geq m_{l}$.
\end{remark}

\textbf{Robustness:} We first flesh out how existing works achieve robustness. Different from the binary systems, the residue number systems are very sensitive to residue errors, where a small error occurring in one residue may cause a large deviation in reconstruction by trivially applying CRT. It is mainly resulted from the non-weighted nature of residue representation.

%On the other hand, if we assume that no modulo operation is involved, i.e., $R_{il} = Y_i + \Delta_{il}$, the maximum reconstruction error we may expect is upper bounded by $\max_{il} |\Delta_{il}|$.

Therefore, the elegant idea applied in Wang's work \cite{closed} is to recover the quotient of $Y_i$ divided by $\Gamma$, i.e., $\lfloor \frac{Y_i}{\Gamma} \rfloor.$ Indeed, once $\lfloor \frac{Y_i}{\Gamma} \rfloor$ is correctly reconstructed, we can escape the restrain of modulo operations and the rest things are trivial to estimate $\langle Y_i \rangle_{\Gamma}$. Geometric explanations can be found in \cite{TSP2018} \cite{towards}. The follow-up works on GRCRT also follow the same idea. The implementation of GRCRT \cite{RCRT-two} \cite{TSP2018} can be simply concluded as two steps:
\begin{enumerate}
     \item First, convert GRCRT to GCRT by constructing new residue set $\hat{\bm{R}}_{l} = \{ \langle\lfloor \frac{Y_i}{\Gamma} \rfloor\rangle_{M_l} , i=1,2,...,N \}$;

     \item Apply GCRT to find the correspondence relationship of elements between $\hat{\bm{R}}_{l}$ and $\bm{Y}_{[1:N]}$ and then reconstruct $\{Y_i\}$.
\end{enumerate}

However, in order to get $\hat{\bm{R}}_{l}$, according to \cite{RCRT-two}, \cite{TSP2018}, it is required that for each $i$, \footnote{We find that the assumptions that all $|\Delta_{il}|<\frac{\Gamma}{4N}$  in \cite{RCRT-two}, \cite{TSP2018} can be relaxed to (\ref{pre}). }
\begin{equation}
\label{pre}
\max_{l} \Delta_{il} - \min_{l} \Delta_{il}<\frac{\Gamma}{2N}=\delta.
\end{equation}
For simplicity, let us assume that the variances $\sigma^2_l$ are the same for each $l$ as $\sigma^2$ temporarily here to ease the analysis. Since the errors $\Delta_{il}$ are i.i.d. Gaussian noises, the probability that (\ref{pre}) holds is
\begin{equation}
\label{minprodensity}
\bigg (\int_{-\infty}^{\infty} p(x) (\Phi(x+2\delta)-\Phi(x))^{L-1} ~ dx \bigg )^{N}
\end{equation}
where $p$ and $\Phi$ are the probability density and cumulative distribution function of a Gaussian $\mathcal{N}(0,\sigma^2)$, respectively. Clearly, (\ref{minprodensity}) can be further upper bounded by
\begin{equation}
\label{pre-upper}
{(\Phi(\delta)-\Phi(-\delta))}^{N(L-1)}
\end{equation}
As $N$ increases, with fixed $\sigma$, (\ref{pre-upper}) decays exponentially in an order of $O(N^2)$. \footnote{$L$ is indeed linear proportional to $N$ since it is required that the value of the lcm of $\frac{L}{N}$ moduli should be bigger than $Y_i$. }

\textbf{Methology:} {In Fig. 1, it provides a more intuitive view with respect to the algebraic structure of residue representation. With modulo operation, the real axis $\mathbb{R}$ is folded and wrapped into a circle, of which the length equals to the modulus. When the moduli are in such a form $m_l = \Gamma M_l$, the following holds:
\begin{equation}
\label{proj_basic}
\mu_{i} = \langle \langle Y_i \rangle_{m_l} \rangle_{\Gamma} = \langle Y_i \rangle_{\Gamma}
\end{equation}
As a property shared by all residues of $Y_i$ modulo $\bm{m}_{[1:L]}$, we term $\mu_i$ the {\em common residue} of $Y_i$. The operation of modulo $\Gamma$ can be viewed as a {\em projection} in residue space, shown in Fig. 1. Obviously, if $\{\mu_{i}\}$ are distinct, they can be used to find the correspondences between $\bm{R}_{[1:N],l}$ and $\bm{Y}_{[1:N]}$. However, with the occurrence of errors, the strategy fails to provide correct determination. However, it inspires us to estimate the correspondences from clustering $\langle R_{il} \rangle_{\Gamma}$. This is the key idea of proposed reconstruction schemes, where we only deal with $\mu_{i}$ instead of searching across $[0,D)$.
On the other hand, CRT plays a role to aggregate the residues across $\bm{m}_{[1:L]}$ to find out the number they represent on the outer circle modulo $\Gamma\prod_{l=1}^LM_l$. The two key operations, projections to the circle modulo $\Gamma$ and CRT, which will frequently appear in the following context, are illustrated in Fig. \ref{proj}.
}

\begin{figure}
\centering
\includegraphics[width=2.2 in,height= 2.5in]{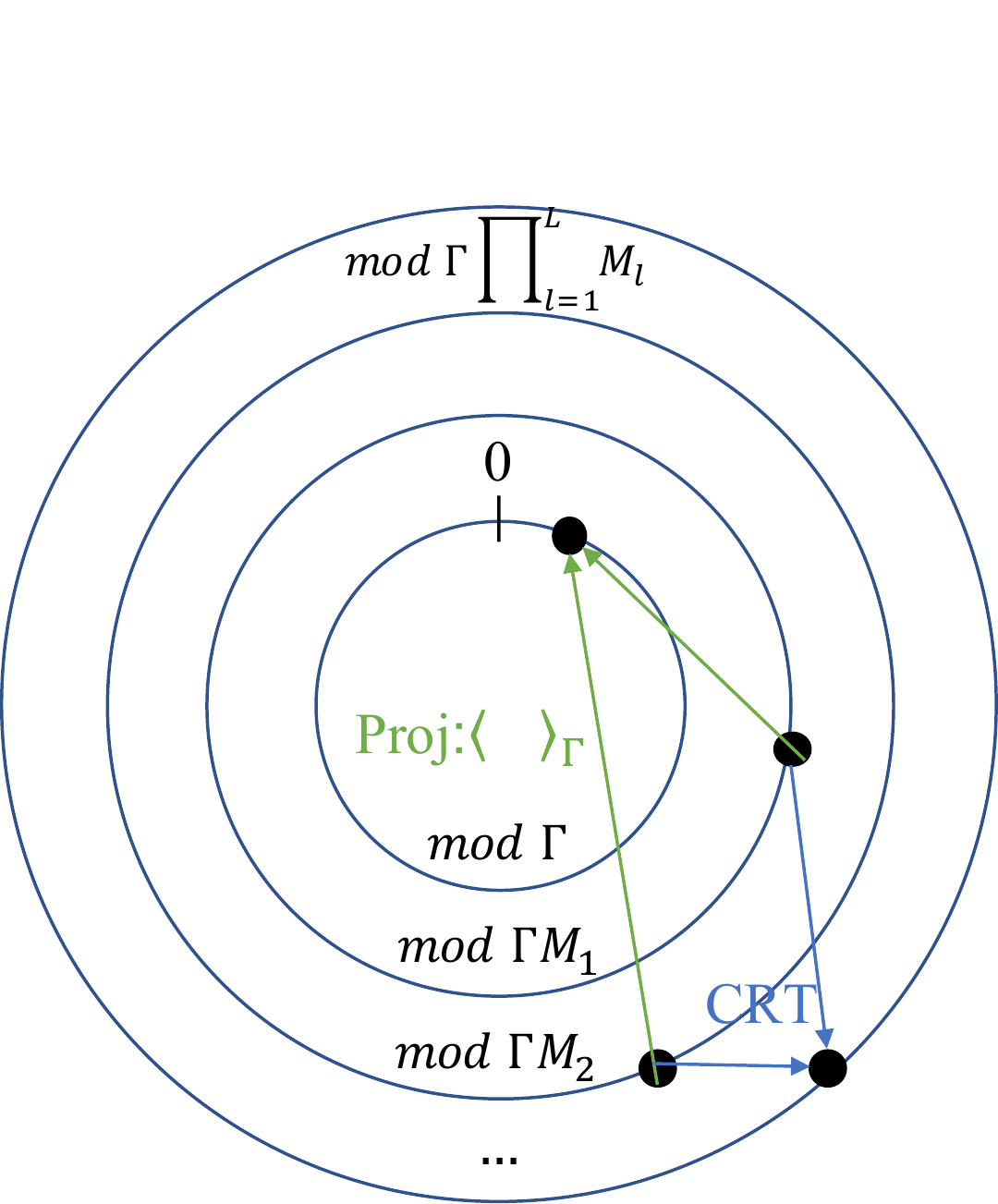}
\caption{Illustration for residue projection and CRT}
\label{proj}
\end{figure}

Throughout the rest of the paper, we will show that in order to achieve the maximal possible dynamic range, all the statistical analyses on $Y_i$ can be elegantly replaced by those of the erroneous common residues, $r_{il}= \langle R_{il} \rangle_{\Gamma}$. For the convenience of readers, all constantly used notations are listed in Table I.

\begin{table}
	\caption{List of Notations}
	\begin{tabular*}{8.8cm}{lll}
		\hline
		Notations & ~~~~~~~~~ Explanation    \\
		\hline
		$L$  & The number of samplings / moduli selected  \\
		$m_l$ & Moduli selected  \\
		%$\mathscr{M}$ & $\mathscr{M}=\{m_l, l=1,2,...,L\}$ \\
		$N$ & The number of real numbers to be reconstructed\\
		$Y_i$ & Real number to be reconstructed\\
		$K_l$ & Permutation variable for each sampling\\
		$R_{il}$ & Raw observations\\
		$\Delta_{il}$ & Gaussian noise in observation\\
		$\mu_i$ & Common residue, residue of $Y_i$ modulo$\Gamma$\\
		$r_{il}$& Residue of observation $R_{il}$ modulo $\Gamma$ \\
		$\hat{Y}_{i}$ & Estimation of $Y_i$ \\
		%$\hat{K}_l$ &Estimation for Permutation variable of each sampling\\
		$\hat{\mu}_{i}$ & Estimation of common residue\\
	\hline
	\end{tabular*}
\end{table}

\section{Algorithm One: Maximum a Posteriori Estimation For Residue Clustering}
\noindent In this section, we will introduce our non-informative prior and describe the problem as a Bayesian statistical model. We further show that under Assumption 1, the MAP of residue clustering is in a semi-closed form and can be determined from $O(NL)$ candidates. Relying on the MAP of residue clustering, it is reduced to $N$ conventional RCRT for a single number.
%\begin{algorithm}
%\caption{Data Generation}
%\textbf{Data Generation for the $l^{th}$ sampling}
%\label{data generation}
%\begin{itemize}
%\item A latent uniform random permutation variable $k_l$ is drawn, where
%\begin{equation}
%k_l = \left\{
% \begin{matrix}
%   1 & 2 & ... & N \\
%   k_l(1) & k_l(2) &... & k_l(N)
% \end{matrix}
%  \right\}
 % \end{equation}
%\item We receive a set of residues $\{ \widetilde{R}_{il}, i=1,2,...,N\}$, where $ \widetilde{R}_{k_l(i)l} = \langle X_i + \Delta_{il} \rangle_{m_l}$, i.e., $ \widetilde{R}_{k_l(i)l}$ is an observation of $X_i$ with noise. Here we assume $\Delta_{il}$ are independently and identically in a normal distribution $N(0,\sigma^2_l)$ and $\sigma_l$ is preset.

%\item With  $\{ \widetilde{R}_{il}, i=1,2,...,N\}$, we further process them to be modulo $\Gamma$, the greatest common divisor(gcd) that all $m_l$ shares, where  $ \widetilde{R}^c_{il} = \langle  \widetilde{R}_{il} \rangle_{\Gamma}$
%\end{itemize}

%\end{algorithm}

%Given $N$ real numbers, $\textbf{Y}_{[1:L]}= \{Y_i, i=1,...,N\}$, they are uniformly distributed in $[0,D)$, where $D$ is dynamic range. There is a set of moduli $\{m_l = \Gamma \times M_l, l = 1, 2, ..., L\}$, which are $L$ sampling rates of samplers. As mentioned in above section, $\Gamma$ is a preset real number, and $M_l$ are relatively co-prime.
For $N$ real numbers, $\textbf{Y}_{[1:N]}= \{Y_i, i=1,...,N\}$ uniformly distributed in $[0,D)$, on achieving the maximal dynamic range, $D$ is set as $D = \Gamma \times \prod_{l = 1}^{L} M_l$.  %and $\sigma_l, l = 1, 2, ..., L$, are designed and preset.
%With the $l^{th}$ sampler, we will observe an unordered set of noisy residues, denoted by $\bm{R}_{[1:N],l}=(R_{1l}, R_{2l}, ... , R_{Nl})^T$. We just know they are the residues of $Y_{[1:N]}$ modulo $m_l$, while we do not know their correspondence relationships.
For brevity, all noisy residues sampled with $L$ samplers are represented by $\bm{R}_{[1:L]}=(\bm{R}_{[1:N],1}, \bm{R}_{[1:N],2}, ... , \bm{R}_{[1:N],L})$. To specify the problem, we introduce $K_l, l =1, 2, .., L$, as a set of i.i.d. $N$-permutation variables, which represents the underlying correspondences between real numbers and residues. It is assumed that the permutation variable $\textbf{K}_{[1:L]}=({K}_1, {K}_2, ... ,{K}_L)$ subjects to uniform distribution.
Under a specific $\textbf{K}_{[1:L]}$, it implies that we assume $\{R_{{K_{l}(i)},l}, l = 1, 2,..., L\}$ are the residues of $Y_i$.
%In the following, we assume $\Delta_{il}, i= 1,2,.., N$ are i.i.d random noise following Gaussian distribution $N(0, \sigma_i^2)$ for each $i$.

We decompose $Y_i$ as $Y_i = k_i \Gamma + \mu_i$, where $\mu_i := \langle Y_i\rangle_{\Gamma}$ denotes the residue of $Y_i$ modulo $\Gamma$, and $k_i$ denotes the corresponding quotient. Since $Y_i$ follows a prior of uniform distribution in $[0, D)$, $k_i$ is an integer random variable uniformly distributed within $\{0, 1, 2,..., \frac{D}{\Gamma}-1\}$, and $\mu_i$ uniformly distributed within $[0,\Gamma)$. Similarly, we decompose $R_{il}$ as $j_{il}  \Gamma +  r_{il}$, where $r_{il} := \langle R_{il} \rangle _{\Gamma}$ denotes residue of $R_{il}$ modulo $\Gamma$, and $j_{il}$ denotes the quotient accordingly. We therefore move all parameters and observations onto a 'smaller (inner) circle' (refer to Fig. 1) modulo $\Gamma$. Accordingly, we estimate $\textbf{K}_{[1:L]}$ with MAP, denoted by $\hat{\textbf{K}}_{[1:L]}$, i.e.,
\begin{equation}
\label{huahua}
\begin{aligned}
\hat{\textbf{K}}_{[1:L]}
& := \arg\max_{\textbf{K}_{[1:L]}}  p(\textbf{K}_{[1:L]}|  \textbf{R}_{[1:L]}  )  \\
& \propto \arg\max_{\textbf{K}_{[1:L]}} p(\textbf{R}_{[1:L]}  |\textbf{K}_{[1:L]}) \\
& \propto \arg\max_{\textbf{K}_{[1:L]}} \int_{Y_1} ... \int _{Y_N} p(\textbf{R}_{[1:L]}| \textbf{Y}_{[1:N]}, \textbf{K}_{[1:L]} ) dY_1...dY_N \\
\end{aligned}
\end{equation}
where $A \propto B$ denotes that for two probability density $A$ and $B$, $A=cB$ for some constant $c$. The complexity of directly solving the above objective function is prohibitively high, where there exist exponential many, $L \times N!$, candidates of $\textbf{K}_{[1:L]}$. On the other hand, given a specific $\textbf{K}_{[1:L]} $, the integration in (\ref{huahua}) can be simplified to calculating the following equation, %where for brevity we remove the variable $\textbf{K}_{[1:L]}$ and assume $\{R_{il}, l=1,2,...,L\}$ corresponds to $Y_i$:

\begin{equation}
\label{deduction2}
\begin{aligned}
&\int_{Y_i} p(\textbf{R}_{[1:L]}| \textbf{K}_{[1:L]} ,Y_i ) dY_i \\
&\propto \int_{0}^{\Gamma} \sum_{k_i = 0}^{\frac{D}{\Gamma}} \prod_{l = 1}^{L} \sum_{j_{K_l(i)l} = -\infty}^{\infty} p(j_{K_l(i)l}  \Gamma +  r_{K_l(i)l} | k_i \Gamma + \mu_i ) d\mu_i  \\
&\propto \int_{0}^{\Gamma} \sum_{k_i = 0}^{\frac{D}{\Gamma}} \prod_{l = 1}^{L}  \sum_{j_{K_l(i)l} = -\infty}^{\infty} \frac{1}{\sqrt{2\pi}\sigma_l} e^{\frac{-(r_{K_l(i)l} - \mu_i + (j_{K_l(i)l} - k_i )\Gamma)^2}{2\sigma_l^2}}    d\mu_i \\
 &\propto \int_{0}^{\Gamma}   \prod_{l = 1}^{L} \sum_{j'_{K_l(i)l} = -\infty}^{\infty}\frac{1}{\sqrt{2\pi}\sigma_l} e^{\frac{-(r_{K_l(i)l} - \mu_i + j'_{K_l(i)j}\Gamma)^2}{2\sigma_l^2}}    d\mu_i
\end{aligned}
\end{equation}
Here, $j_{K_l(i)l}$ enumerates all integers in $\mathbb{Z}$, so does $j'_{K_l(i)j} = j_{K_l(i)l} - k_i$. It is noted that when $L = 1$, we can get:
\begin{equation}
\label{L_1}
\begin{aligned}
 &\sum_{j'_{K_l(i)l} = -\infty}^{\infty} \int_{0}^{\Gamma} \frac{1}{\sqrt{2\pi}\sigma_l} e^{\frac{-2(r_{K_l(i)l} - \mu_i + j'_{K_l(i)l}\Gamma)^2}{2\sigma_l^2}}    d\mu_i \\
&= \int_{-\infty}^{\infty} \frac{1}{\sqrt{2\pi}\sigma_l} e^{\frac{-2(r_{K_l(i)l} - \mu_i)^2}{2\sigma_l^2}}    d\mu_i
\end{aligned}
\end{equation}
This motivates us to think about whether we can remove the product term on $l$ in (\ref{deduction2}) and simplify it into a closed-form formula as (\ref{L_1}) under some mild assumptions. In the following, we introduce Assumption 1, under which a polynomial time algorithm is creatively proposed to deterministically derive the MAP estimation for $\textbf{K}_{[1:L]}$. We start from introducing some notations for {\em noise distributing intervals}: for each $i = 1, 2,..., N$, we define an clockwise interval $I_i$ as $I_i = [ \mu_i+ \min_{l} \Delta_{il},  \mu_i+ \max_{l} \Delta_{il} ]$, i.e., starting from $\mu_i+ \min_{l} \Delta_{il}$ to $\mu_i+ \max_{l} \Delta_{il}$ clockwise, which are illustrated in Fig. \ref{circle-interval}. In addition, let $|I_i|$ denote the length of the directed interval $I_i$, i.e., $|I_i|:=\max_{l} \Delta_{il} - \min_{l} \Delta_{il}$.

\begin{figure}
\centering
\includegraphics[width=2.87 in,height=1.66 in]{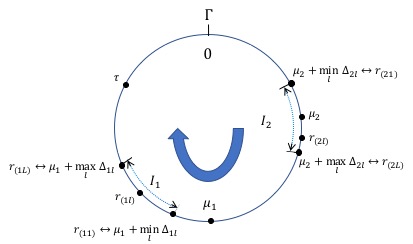}
\caption{Illustration for the noise interval}
\label{circle-interval}
\end{figure}

\begin{figure}
\centering
\includegraphics[width=2.795 in,height=0.572in]{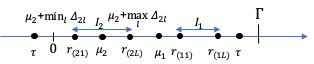}
\caption{Illustration for straightened circle cut at $\tau$  }
\label{line}
\end{figure}

\begin{cond}
\label{cond}
There exists some point $\tau$ on the circle modulo $\Gamma$ such that it is not within any interval $I_i$ and $|I_i| < \frac{\Gamma}{2}$ for $i=1,2,...,N$.
\end{cond}

\begin{remark}
As illustrated in Fig. \ref{circle-interval}, $\tau$ can be an arbitrary point on the circle which is not within any intervals $I_i$. On the other hand, even if $\bm{K}_{[1:L]}$ is determined and the problem is simplified to $N$ independent RCRT for a single number, we still need further limitations on $\Delta_{il}$ to guarantee successful reconstructions. Robustness is proved to be achieved in \cite{closed}, \cite{sp2017} when $|I_i| < \frac{\Gamma}{2}$ for $i=1,2,...,N$. That is the reason why we assume $|I_i| < \frac{\Gamma}{2}$ for each $i$ in Assumption 1.
\end{remark}

Assumption 1 provides the convenience in analysis where the circle can be virtually cut at point $\tau$, and straightened into a line where the order of $\Delta_{il}$ for each $i$ is still preserved. Fig. \ref{line} is an illustration for the above operation if we cut the circle in Fig. \ref{circle-interval} at $\tau$ and straighten it to a line. The order of $r_{il}$ on the line corresponds to the order $\Delta_{il}$ accordingly in an ascending order. Specifically, for each $i$, we denote $\{r_{(il)}, l = 1, 2,..., L\}$ as a clockwise order statistic of $\{r_{il}, l = 1, 2,..., L\}$. Here, $(il)$ denotes a permutation on the index $\{l, l = 1,2,...,L\}$ for each $i$, such that $r_{(il)}$ is $l^{th}$ element of $\{r_{il}, l = 1, 2,..., L\}$ clockwise distributed on the circle starting from $\tau$, illustrated in Fig. \ref{circle-interval}.

\begin{lm}
\label{cutting}
For each $i$, the errors $\{\Delta_{(il)}\}$ are corresponding to the subsequence $r_{(i1)},r_{(i2)}, ...,r_{(iL)}$, which are in ascending order, i.e., $\{\Delta_{(il)}\}$ are the order statistic.
\end{lm}

\begin{proof}
According to Assumption 1, $\tau$ is not within any $I_i$. Therefore, $I_i$, the directed interval defined, is clockwise distributed starting from $\langle \mu_i + \min_l \Delta_{il} \rangle_{\Gamma}$ to $\langle \mu_i + \max_l \Delta_{il} \rangle_{\Gamma}$ and $\tau$ is in the complementary part, $[0,\Gamma) / I_i$. It is obvious that $r_{(i1)}$ is closest to $\tau$ in counterclockwise direction, while $r_{(iL)}$ is the closest one in clockwise direction among all $r_{i[1:L]}$. For each $i$, $r_{(il)}$ is exactly in the order, so corresponding $\{\Delta_{(il)}\}$ are arranged in ascending order as well.
\end{proof}

Lemma 1 shows that if $\tau$ is known, starting from $\tau$, for each $i$, the clockwise order of $\{r_{il}\}$ on the circle is exactly the order of $\{\Delta_{il}\}$ in an ascending order accordingly. In order to intuitively understand the relative positions of $\{r_{il}\}$, we convert the distribution of them on a circle to the one on an axis (cutting the circle at $\tau$ and stretching it into a line), which will ease the following analysis. To this end, we give the following definition and lemma.

\begin{df}
When $0 \leq \tau \leq \min {r}_{{(il)}}$ or $\max {r}_{{(il)}} \leq \tau <\Gamma$, for $i=1,2,...,N$ and $l=1,2,...,L$,

\begin{equation}
\label{shift}
\tilde{r}_{{il}} = {r}_{{il}}
\end{equation}

Otherwise,
\begin{equation}
\begin{cases}
\tilde{r}_{{il}} = {r}_{{il}},  ~~~~~when~~ {r}_{{il}} \leq \tau \\
\tilde{r}_{{il}} = {r}_{{il}}-\Gamma, ~~~~ when~~ {r}_{{il}} > \tau  \\
\end{cases}
\end{equation}
\end{df}

%Back to our problem with Condition 1, the optimization is the maximum a posteriori probability (MAP) estimate of $K_l$.
%\begin{equation}
%\label{obj}
%\begin{aligned}
%& \arg_{\textbf{K}} \max p( K_1, K_2, ... ,K_L| \textbf{R}_{[1:L]}, Cond ~1) \\
%& = \arg_{\textbf{K}}\max \frac{p( K_1, K_2, ... ,K_L, r_{il}, Cond ~1)}{p(r_{il}, Cond ~1) }\\
%& \propto   \arg_{\textbf{K}} \Pr(K_1, K_2, ... ,K_L, r_{il}, Cond ~1) \\
% & = \Pr(s_{il} | Cond ~1,  k_1, k_2, ... ,k_L ) \Pr( Cond ~1 , k_1, k_2, ... ,k_L )
%\end{aligned}
%\end{equation}

%In the following, we firs prove the following lemma to further simplify the optimization of (\ref{obj})
%\begin{lm}
% $\Pr( Cond ~1,  k_1, k_2, ... ,k_L )$ is a constant independent of $s_{il}$.
%\end{lm}

%\begin{proof}
%For a permutation $k_i$ on $\{1,2,...,N\}$, there also exists an alternative definition on $k_i$ such that it is a permutation on the index that the samples $s_{il} \in [0,\Gamma)$ are sorted in an ascending order. Since we can always add a rotation on the circle without changing their relative positions of $s_{il}$.  ... Due to the symmetry,  $\Pr( Condition ~1 |  k_1, k_2, ... ,k_L )$ are all equal to some constant.
%We show that $\Pr( Cond ~1,  k_1, k_2, ... ,k_L ) \in \{0,1\}$.
%\end{proof}

% Thus we only need to focus on $\Pr(s_{il}, Condition ~1 |  k_1, k_2, ... ,k_L )$.
%Under Assumption 1, we have to further define \textit{proper classification}, where the specific statements are postponed but the following lemma provides the motivation and some insights.
\begin{lm}
Under Assumption 1, given $\{r_{il}\}$ and $\{K_l\}$ for $i=1,2,...,N$ and $l=1,2,...,L$, $I_i$ can be uniquely determined.
\end{lm}

\begin{proof}
All samples $\{r_{il}\}$ are divided into $N$ subsets according to $\{K_l\}$, each of which includes the $L$ error residues of $Y_i$, represented by $\{r_{K_l(i)l}, l=1,2,...,L\}$. When Assumption 1 holds, there should exist a clockwise directed interval over the circle starting from $r_{K_{l_1}(i)l_1}$ and ending at $r_{K_{l_2}(i)l_2}$ for some $l_1 \not = l_2 \in \{1,2,...,L\}$ such that the length is smaller than  $\frac{\Gamma}{2}$. All the remaining samples $\{r_{K_l(i)l}, l=1,2,...,L, l\not =l_1,l_2\}$ lie in the interval $I_i$. Clearly, $r_{K_{l_2}(i)l_2}$ is clockwise neighboring to $r_{K_{l_1}(i)l_1}$. We claim such an interval is unique. Otherwise, we assume there are two indices $l_3$ and $l_4$ such that the clockwise directed interval, $I'_i$, starting from $r_{K_{l_3}(i)l_3}$ to $r_{K_{l_4}(i)l_4}$ also has a length smaller than $\frac{\Gamma}{2}$, which contains $\{r_{K_l(i)l}, l=1,2,...,L, l\not =l_3,l_4\}$. Thus, the interval $I'_i$ includes the complement part of $I_i$. Since $|I_i|$ is smaller than $\frac{\Gamma}{2}$, therefore, $|I'_i| \geq \Gamma - |I_i| >\frac{\Gamma}{2}$, which incurs a contradiction. Thus, our claim holds.
\end{proof}

To proceed from Lemma 2, we use $A1$ standing for Assumption 1 for simplicity, shown in the following formulas. We further modify our objective function as follows:
\begin{align}
\begin{split}
\hat{\textbf{K}} _{[1:L]}
& := \arg\max_{\textbf{K}_{[1:L]}}  p(\textbf{K}_{[1:L]} |  \textbf{R}_{[1:L]}, A1)  \\
& \propto \arg\max_{\textbf{K}_{[1:L]}} p(\textbf{K}_{[1:L]} |  \textbf{R}_{[1:L]}, A1)  \times p(\textbf{R}_{[1:L]}|A_1) \\
& = \arg\max_{\textbf{K}_{[1:L]}} p(\textbf{R}_{[1:L]}  | A1, \textbf{K}_{[1:L]}) \times p(\textbf{K}_{[1:L]}|A_1) \\
\end{split}
\end{align}

In the following, we prove Assumption 1 is independent of permutation $\textbf{K}_{[1:L]}$. For any $\textbf{K}_{[1:L]}$, we have
\begin{align}
\begin{split}
\Pr (A1 | \textbf{K}_{[1:L]}) = \int_{ \textbf{R}_{[1:L]}} p(A1, \textbf{R}_{[1:L]}| \textbf{K}_{[1:L]}) d \bm{R}_{[1:L]} \\
= \int_{\bm{R}_{[1:L]}} p(   A1 |\textbf{K}_{[1:L]}  , \textbf{R}_{[1:L]}  ) \times p(\textbf{R}_{[1:L]} | \textbf{K}_{[1:L]} )d \bm{R}_{[1:L]} \\
= \int_{\bm{R}_{[1:L]}} p(  A1 |\textbf{K}_{[1:L]} , \bm{R}_{[1:L]}  ) \times p(\bm{R}_{[1:L]} ) d \bm{R}_{[1:L]}
\end{split}
\end{align}
Also,
\begin{equation}
\begin{aligned}
&\Pr (A1) =\int_{\bm{R}_{[1:L]}} \sum_{\textbf{K}_{[1:L]}} p(A1, \textbf{K}_{[1:L]}, \bm{R}_{[1:L]}) \times p(\textbf{K}_{[1:L]}) d\bm{R}_{[1:L]}\\
&=\sum_{\textbf{K}_{[1:L]}} \int_{\bm{R}_{[1:L]}} p(A1|\textbf{K}_{[1:L]},\bm{R}_{[1:L]}) \\
& ~~~~~~~~ ~~~~~~~~~~~\times p(\bm{R}_{[1:L]}|\textbf{K}_{[1:L]}) \times p(\textbf{K}_{[1:L]}) d\bm{R}_{[1:L]}
\end{aligned}
\end{equation}

Since we assume $\textbf{K}_{[1:L]}$ are uniformly distributed, it suffices to show independence that $\int_{R} p(A1|\textbf{K}_{[1:L]},\bm{R}_{[1:L]}) \times p(\bm{R}_{[1:L]}|\textbf{K}_{[1:L]}) d\bm{R}_{[1:L]}$ remains constant across all $\textbf{K}_{[1:L]}$. On the other hand, as the conditional probability density of $p(\bm{R}_{[1:L]}|\textbf{K}_{[1:L]})$ is a normal distribution, we know $\int_{\bm{R}_{[1:L]}} p(A1|\textbf{K}_{[1:L]},\bm{R}_{[1:L]}) \times p(\bm{R}_{[1:L]}|\textbf{K}_{[1:L]}) d\bm{R}_{[1:L]}$ is constant across all $\textbf{K}_{[1:L]}$. Thus, Assumption 1 is independent of permutation $\textbf{K}_{[1:L]}$.

 If $\textbf{K}_{[1:L]}$ is a correct residue classification, assuming that the $cutting~point$ is $\tau$ and following the notations given in Definition 1, a closed form of (\ref{huahua}) can be derived as follows.

\begin{lm} When $\Pr(\textbf{K}_{[1:L]},\textbf{R}_{[1:L]}, A1) \not =0$,
 \begin{equation}
 \begin{aligned}
 \label{MAP}
 \Pr(\textbf{R}_{[1:L]} | A1, & \textbf{K}_{[1:L]} )  = \Pr(\textbf{r}_{[1:L]} | A1, \textbf{K}_{[1:L]} )    \\
& \propto \prod_{i=1}^{N} \int_{-\infty}^{\infty} e^{-\sum_{l=1}^{L} w_l (x-\tilde{r}_{K_l(i)l})^2} dx
  \end{aligned}
 \end{equation}
 \end{lm}

where $w_l$ is the weight determined by $\sigma_l$, i.e., $w_l =\frac{1}{2\sigma^2_l}.$
 \begin{proof}
First we need to clarify, given $\textbf{r}_{[1:L]} $ and $\textbf{K}_{[1:L]}$, if there are multiple candidates of cutting point $\tau$, the value of (\ref{MAP}) is invariant to different selections of $\tau$. For a fixed $i$ and different possible $\tau \not \in I_i$, the relative positions of $\{\tilde{r}_{K_l(i)l}\}$ do not change, as proved in Lemma 1. The only difference is that there may be a uniform shift on $\{\tilde{r}_{K_l(i)l}\}$, i.e., two different cutting points may result in two different groups $\{\tilde{r}_{K_l(i)l}\}$ and $\{\tilde{r}'_{K_l(i)l}\}$ whereas $\tilde{r}_{K_l(i)l}-\tilde{r}'_{K_l(i)l}$ equals a constant: a multiple of $\Gamma$. However, replacing $\{\tilde{r}_{K_l(i)l}\}$ with $\{\tilde{r}'_{K_l(i)l}\}$ in (\ref{MAP}), the value of (\ref{MAP}) does not change due to the integral on $x$ along the $\mathbb{R}$.

With lemma 2 and Assumption 1, for each $i$, given some $\mu_i \in [0,\Gamma)$, if we sort $\{\tilde{r}_{K_l(i)l}\}$ in an ascending order, accordingly $\{\mu_i+\Delta_{K_l(i)l}\}$ are also sorted ascendingly. Therefore, the errors $\{ \Delta_{K_l(i)l}, l=1,2,...,L\}$ must be in a form $\{\mu_i-\tilde{r}_{K_l(i)1}+g\Gamma,\mu_i-\tilde{r}_{K_l(i)2}+g\Gamma, ... ,\mu_i-\tilde{r}_{K_l(i)L}+g\Gamma\}$, for some $g \in \mathbb{Z}$. On the other hand, since $\mu_i$ are i.i.d. uniformly distributed in $[0,\Gamma)$, we can conclude that, for each $i$, under the residue classification $\bm{K}_{[1:L]}$, the probability density of $r_{il}$ is proportional to $\int_{-\infty}^{\infty} e^{-\sum_{l=1}^{L} w_l (x-\tilde{r}_{K_l(i)l})^2} dx$. Due to the independence of $\mu_i$, (\ref{MAP}) follows.
 \end{proof}

Furthermore, we show in the following that there exists an efficient scheme to determine the optimal solutions of  (\ref{huahua}).  Before proceeding, we introduce the following notations for clarity. For any given $\tau \in [0,\Gamma)$, let $\gamma^{\tau}_{(i)l}$, $l \in \{1,2,...,L\}$, denote the $i^{th}$ item of $\{\tilde{r}_{il}, i=1,2,...,N\}$ sorted in ascending order.
 \begin{thm}
 The MAP estimation of $\bm{K}_{[1:L]}$ under Assumption 1 is to determine a cutting point $\tau \in [0,\Gamma)$ such that
 \begin{equation}
 \label{thm2}
 \arg \max_{\tau}  \sum_{i=1}^{N} [\frac{(\sum_{l=1}^{L}\tilde{r}_{K_l(i)l} w_l)^2}{\sum_{l=1}^{L} w_l}- \sum_{l=1}^{L} w_l \tilde{r}^2_{K_l(i)l}]
 \end{equation}
and the optimal clustering strategy is to group $\{\gamma^{\tau}_{(i)l}, l=1,2,...,L\}$: i.e., clustering the $i^{th}$ largest elements among each set $\{ \tilde{r}_{[1:N]l}\}$ together for each $i$.
 \end{thm}

\begin{proof}
For $\textbf{K}_{[1:L]} = (K_1,K_2,...,K_L)$ of any correct residue classification, (\ref{MAP}) can be further simplified as,
\begin{equation}
\begin{aligned}
\label{combine}
 & \prod_{i=1}^{N} \int_{-\infty}^{\infty} e^{-[ (\sum_{l=1}^{L} w_l) x^2 -2\sum_{l=1}^{L}\tilde{r}_{K_l(i)l} w_lx+ \sum_{l=1}^{L} w_l \tilde{r}^2_{K_l(i)l}]} dx\\
&=\prod_{i=1}^{N} \int_{-\infty}^{\infty} exp [- (\sum_{l=1}^{L} w_l )(x-\frac{\sum_{l=1}^{L}\tilde{r}_{K_l(i)l} w_l}{\sum_{l=1}^{L} w_l})^2  \\
&+ \frac{(\sum_{l=1}^{L}\tilde{r}_{K_l(i)l} w_l)^2}{\sum_{l=1}^{L} w_l}- \sum_{l=1}^{L} w_l \tilde{r}^2_{K_l(i)l}]  dx\\
& \propto \sum_{i=1}^{N} [\frac{(\sum_{l=1}^{L}\tilde{r}_{K_l(i)l} w_l)^2}{\sum_{l=1}^{L} w_l}- \sum_{l=1}^{L} w_l \tilde{r}^2_{K_l(i)l}]\\
% & = \sum_{i=1}^{N} \sum_{l=1}^{L} \frac{\hat{s}^2_{k_l(i)l} w^2_l}{\sum_{l=1}^{L} w_l} +
\end{aligned}
\end{equation}

Given $\tilde{r}_{il}$, $\sum_{i=1}^{N}\sum_{l=1}^{L} w_l \tilde{r}^2_{K_l(i)l}$ in (\ref{combine}) is a constant. Then, we only need to focus on
\begin{equation}
\label{final}
 \sum_{i=1}^{N}  (\sum_{l=1}^{L}\tilde{r}_{K_l(i)l} w_l)^2
\end{equation}
For the rest, we first prove that the residue cluster following the rule of grouping $ C^{\tau}_i = \{\gamma^{\tau}_{(i)1},\gamma^{\tau}_{(i)2},...,\gamma^{\tau}_{(i)L}\}$ for each $i$ achieves the maximum value of (\ref{final}). This is a generalization of the following inequality. For two pairs of numbers $a_1 \leq a_2$ and $b_1 \leq b_2$, we have
\begin{equation}
{(a_1+b_1)}^2+{(a_2+b_2)}^2 \geq {(a_1+b_2)}^2+{(a_2+b_1)}^2
\end{equation}
In general, for two sequences, $\{\gamma^{\tau}_{(i)1},\gamma^{\tau}_{(i)2},...,\gamma^{\tau}_{(i)L}\}$ and $\{\gamma^{\tau}_{(j)1},\gamma^{\tau}_{(j)2},...,\gamma^{\tau}_{(j)L}\}$, both of which are sorted in non-decreasing order, the rearrangement inequality \cite{inequalities} tells that
\begin{equation}
\label{inequ}
\begin{aligned}
\gamma^{\tau}_{(K(1))1}\gamma^{\tau}_{(1)2} &+ \gamma^{\tau}_{(K(2))1}\gamma^{\tau}_{(2)2}+...\gamma^{\tau}_{(K(N))1}\gamma^{\tau}_{(N)2}  \\
& \leq \gamma^{\tau}_{(1)1}\gamma^{\tau}_{(1)2} + \gamma^{\tau}_{(2)1}\gamma^{\tau}_{(2)2}+...\gamma^{\tau}_{(N)1}\gamma^{\tau}_{(N)2}
\end{aligned}
\end{equation}
where $K$ can be any permutation on $\{1,2,...,N\}$. Said another way, the maximum value of (\ref{final}) is achieved when the order is preserved. According to (\ref{inequ}), the optimal value of (\ref{final}) is obtained following the clustering strategy claimed: we group the $i^{th}$ largest elements among $\{\tilde{r}_{[1:N]l}\}$ for each $l$ together.

Since there are $NL$ candidate cutting points, where one may select $\tau = r_{il}$ for each $i$ and $l$, what we prove above presents the local optimal classification strategy for a $\tau$. Therefore, in the worst case, by enumerating all the $NL$ candidate cutting points, we can find the final optimal solution to (\ref{huahua}).
\end{proof}
To conclude, the complexity of computing the MAP for residue clustering under Assumption 1 is reduced to find out the optimal $\tau$ from $NL$ candidates $\{r_{il}\}$.  We conclude the proposed algorithm as follows.

\begin{algorithm}
\caption{Conditional MAP Estimation of Classification}
\textbf{Input:} Given moduli $m_l=\Gamma M_l$ and the residues observed $R_{il}, i=1,2,...,N, l=1,2,...,L$.

1. Calculate $ {r}_{il} = \langle  R_{il} \rangle_\Gamma$;

2. Calculate $\widetilde{r}_{il}$ according to Definition 1.

3. Derive the permutation $K_l$ according to Theorem 1, i.e., find the best $\tau$.

4. Apply the conventional RCRT for a single number to get $\{\hat{Y}_i\}$.

\textbf{Output:} $\hat{Y}_i$, $i=1,2,...,N$.
\end{algorithm}

\begin{example}
Consider $m_1 = 5 \times 2$, $m_2 = 5 \times 3$, $Y_1 =11$ and $Y_2 = 18$.  For simplicity $w_1 = w_2 = 1$, i.e., noises are in a same level perturbing the samples with sampling rate $m_1$ and $m_2$. Two observations are obtained $\bm{R}_1=\{1,9\}$ and $\bm{R}_2=\{10,3\}$. Accordingly, one can derive $\bm{r}_1=\{1,4\}$ and $\bm{r}_2=\{0,3\}$. Under Assumption 1, $\tau$ can be selected from $\{0,1,3,4\}$. Here we specify the two cases where $\tau=1$ and $\tau=3$. When $\tau=1$, recalling Definition 1, $\tilde{\bm{r}}_1=\{1,-1\}$ and $\tilde{\bm{r}}_2=\{0,-2\}$. According to Theorem 1, the clustering strategy is to group the smallest ones, i.e., $\{-1,-2\}$, and group the largest ones, i.e., $\{1,0\}$, in $\tilde{\bm{r}}_1$ and $\tilde{\bm{r}}_2$, respectively. In this scenario,  the loss function in (\ref{thm2}) equals $-1$. When $\tau =3$, $\tilde{\bm{r}}_1=\{1,-1\}$ and $\tilde{\bm{r}}_2=\{0,3\}$. Similarly, we group $\{-1,0\}$ and $\{1,3\}$ together. The value of loss function in (\ref{thm2}) is $-\frac{5}{2}$ then. Similarly, when we set $\tau=4$ or $\tau=0$, the values of (\ref{thm2}) are the same: $-1$, which achieves the maximal of (\ref{thm2}) among the four cases of $\tau$. \footnote{$\tau=0,2,4$ are all cutting points in this example and lead to the same clustering.} Thus, we find the MAP of clustering by grouping $\{1,10\}$ and $\{9,3\}$ as the residues of $Y_1$ and $Y_2$, respectively. The following reconstruction is applying RCRT for a single number on the two residues sets, respectively.
\end{example}

\section{Algorithm Two: Bayesian Wrapped Gaussian Mixture Model and Two-Step Maximization Fast Algorithm}
\noindent In last section, we studied a conditional MAP of residue clustering. It is noted that after the permutations $\bm{K}_{[1:L]}$ are estimated, we still need to apply conventional RCRT for a single number to derive the final reconstruction of $Y_i$ \cite{closed}. It is therefore an interesting question that whether we can estimate both permutation $\textbf{K}_{[1:L]}$ and $\bm{Y}_{[1:N]}$ at the same time.

In this section, we develop a two-step searching algorithm to figure out the estimations of both. Coming with a slight compromise in computational complexity, the method proposed in this section can achieve stronger robustness compared to Algorithm 1. As mentioned above, if we further consider the problem on the 'small circle' modulo $\Gamma$, we would find it similar to the Gaussian Mixture Model (GMM), where the differences lie on the wrapped gaussian distribution for noisy $R_{il}$ and prior knowledge with respect to the sample generation. Inspired with the techniques to solve GMM, in the following, we will treat both $\bm{{Y}}_{[1:N]}$ and $\textbf{K}_{[1:L]}$ as the targets of estimation instead of $\textbf{K}_{[1:L]}$ only, and develop a MAP estimation for both variables at the same time. From (\ref{deduction2}), it is not hard to observe that $ \Pr (\textbf{K}_{[1:L]}, \bm{Y}_{[1:N]} |  \textbf{R}_{[1:L]}  ) = \Pr (\textbf{K}_{[1:L]}, \bm{\mu}_{[1:N]} |  \textbf{r}_{[1:L]} )$, which implies that we only need to deal with the MAP of $(\textbf{K}_{[1:L]},\bm{\mu}_{[1:N]})$ instead.

To this end, the objective function becomes

\begin{equation}
\begin{aligned}
&\{\hat{\textbf{K}}_{[1:L]}, \bm{\hat{\mu}}_{[1:N]}\}\\
& :=\arg\max_{\textbf{K}_{[1:N]}, \bm{\mu}_{[1:N]}}  \Pr (\textbf{K}_{[1:L]}, \bm{\mu}_{[1:N]} |  \textbf{r}_{[1:L]}  )  \\
&\propto \arg\max_{\textbf{K}_{[1:N]}, \bm{\mu}_{[1:N]}} \Pr (\textbf{r}_{[1:L]}  |  \textbf{K}_{[1:L]}, \bm{\mu}_{[1:N]}) \\
\end{aligned}
\end{equation}

We propose an iterative method to solve the above equation. It proceeds as follows: after initializing $\bm{\mu}_{[1:N]}^{(0)}$, for $(t+1)^{th}$ iteration,\begin{itemize}
	\item Step one: given $\bm{\mu}_{[1:N]}^{(t)}$, deducing:
	\begin{equation}
	\label{StepOne}
	{\textbf{K}}_{[1:L]}^{(t+1)} = \arg\max_{\textbf{K}_{[1:L]}} \Pr (\textbf{r}_{[1:L]}  | {\textbf{K}}_{[1:N]}, \bm{\mu}_{[1:N]}^t)
	\end{equation}
	\item Step two: given $\textbf{K}_{[1:L]}^{(t +1)}$, deducing:
	\begin{equation}
	\bm{{\mu}}_{[1:N]}^{(t+1)} = \arg\max_{\bm{\mu}_{[1:N]}} \Pr (\textbf{r}_{[1:L]}  | \textbf{K}_{[1:L]}^{(t+1)}, \bm{\mu}_{[1:N]})
	\end{equation}
\end{itemize}

In the remaining part of this section, we will propose a fast algorithm to solve each step and prove that it will converge to stead state.
%In the followings, we will always initialize $\bm{\mu}_{[1:N]}^{(0)}$ as the $l^{th}$ set of observation $\bm{r}_{[1:N],l} = \{r_{1,l}, r_{2,l},...,r_{N,l}\}$, where $l$ is randomly drawn from $\{1,2,...,L\}$. As the latent variable is a permutation, this initialization is relatively good guess.
We start from deducing a fast algorithm for step one. Similar to equation (\ref{deduction2}), we have

\begin{equation}
\label{target two}
\begin{aligned}
&\Pr(\textbf{r}_{[1:L]}| \bm{K}_{[1:L]}, \bm{\mu}_{[1:N]} ) \\
& \propto \prod_{l = 1}^{L} \prod_{i = 1}^{N} \sum_{j_{K_{l}(i)l} = -\infty}^{\infty} p(j_{K_{l}(i)l}  \Gamma +  r_{K_{l}(i)l} | k_i \Gamma + \mu_{i}  )  \\
&\propto \prod_{l = 1}^{L} \prod_{i = 1}^{N} \sum_{j_{K_{l}(i)l} = -\infty}^{\infty} \frac{1}{\sqrt{2\pi}\sigma_l} e^{\frac{-(r_{K_{l}(i)l} - \mu_{i}  + (j_{K_{l}(i)l} - k_i )\Gamma)^2}{2\sigma_l^2}} \\
 &\propto \prod_{l = 1}^{L} \prod_{i = 1}^{N} \sum_{j'_{K_{l}(i)l} = -\infty}^{\infty}\frac{1}{\sqrt{2\pi}\sigma_l} e^{\frac{-(r_{K_{l}(i)l} - \mu_{i}  + j'_{K_{l}(i)j}\Gamma)^2}{2\sigma_l^2}}
 \end{aligned}
 \end{equation}

Since $K_l$ are independently and randomly distributed, we may simplify (\ref{target two}) to find an optimal $K_l^{(t+1)}$ for each $l$:

\begin{equation}
\label{target two further}
\begin{aligned}
&K_l^{(t+1)} := \arg\max_{K_l} \prod_{i = 1}^{N}\sum_{j'_{il} = -\infty}^{\infty} e^{\frac{-(r_{il} - \mu^t_{K_{l}}(i) + j'_{il}\Gamma)^2}{2\sigma_l^2}}
\end{aligned}
\end{equation}

In general, since (\ref{target two further}) is hard to solve, here we apply the approximation method used in \cite{mle1}. We define $d_{\Gamma} (a,b) := \min_{j \in \mathbb{Z}} |a-b+j\Gamma |$ as the distance of any two real numbers $a$ and $b$. When the $\sigma_l^2$ is much smaller than $\Gamma$, (\ref{target two further}) can be approximated as
\begin{equation}
\label{circle1}
K_l^{(t+1)} =\arg\min_{K_l}  \sum_{i=1}^{N} d^2_{\Gamma} (r_{il}, \mu^t_{K_l(i)}),
\end{equation}
since $exp(-{\frac{d^2_{\Gamma} (r_{K_l(i)l}, \mu_{i})^2}{2\sigma_l^2}})$ dominates the term $\sum_{j'_{K_l(i)l} = -\infty}^{\infty} exp({\frac{-(r_{K_l(i)l} - \mu_{i} + j'_{K_l(i)l}\Gamma)^2}{2\sigma_l^2}})$.
Although solving (\ref{circle1}) is seemingly of exponential complexity, we will show there exists an $O(N)$-time algorithm in the following theorem.  Let $r_{(i)l}$, $i=1,2,...,N$, denote the $i^{th}$ element of increasingly sorted sequence $\{ r_{(1)l}, r_{(2)l}, ... ,r_{(N)l} \}$ of $\bm{r}_{[1:N],l}$. Similarly, ${\mu}^t_{[i]}$ denotes the $i^{th}$ element of increasingly sorted sequence $\{ {\mu}^t_{[1]}, {\mu}^t_{[2]}, ... ,{\mu}^t_{[N]}\}$ of ${\bm{\mu}}^t_{[1:N]}$.

\begin{thm}
	\label{circlemin1}
	There exists some $\zeta \in \{1,2,...,N\}$ such that the following matching strategy: $(r_{(\langle i+\zeta\rangle_N ) l }, \hat{\mu}_{[i]})$, $i=1,2,...,N$, minimize (\ref{circle1}).
\end{thm}

\begin{proof}
	Let $\omega_{i}$ and $\theta_{i}$ within $[0,2\pi)$ denote the angles of $r_{(i)l}$ and ${\mu}^t_{[i]}$ distributed on the 'small circle' modulo $\Gamma$, respectively. Therefore, both $\{\omega_{i}\}$ and $\{\theta_{i}\}$ are also in an ascending order. Correspondingly, the difference of angles between any pair $({\mu}^t_{[i_1]}, r_{(i_2)l})$ is proportional to $\min \{|\omega_{i_2}-\theta_{i_1}|, 2\pi-|\omega_{i_2} - \theta_{i_1}|\}$. To give its geometrical interpretation, consider two concentric circles, as shown in Fig. \ref{concentric}, where $\{{\mu}_{[i]}^t\}$ are distributed on the outer circle and $\{r_{(i)l}\}$ are on the inner one. We define $\xi_{i_1, i_2}$ to represent an angle difference from ${\mu}^t_{[i_1]}$ to $r_{(i_2)l}$ in a counterclockwise direction as positive and otherwise as negative, which is in $(-\pi,\pi]$, illustrated in Fig. \ref{concentric}. Suppose that there is a line connecting each pair $({\mu}^t_{[i_1]}, r_{(i_2)l})$ under the optimal choice, which is also denoted by $\xi_{i_1, i_2}$, {{\em we will prove that those $N$ lines do not intersect each other}}.
	
	We prove it by contradiction. Suppose there exist $i_1$ and $i_2$ such that the two lines $\xi_{i_1, K(i_1)}$ and $\xi_{i_2, K(i_2)}$ cross each other, without loss of generality, we set $\theta_{i_1}=0$, as two concentric circles rotating simultaneously will not affect the distributions of $\theta_{i}$ and $\omega_{i}$ on the two circles. Also we set $\theta_{i_2}\in (0,\pi)$. Note that when $\theta_{i_2}=\pi$, it is impossible to result in crossing. Therefore, when two lines cross, it falls into one of the following four cases:
	\begin{itemize}
		\item 1) $\xi_{i_1, K(i_1)}\geq 0$ and $\xi_{i_2, K(i_2)}\geq 0$, i.e., $0 \leq \theta_{i_1}<\omega_{K(i_2)}<\omega_{K(i_1)} \leq \pi$ and $0<\theta_{i_2}<\omega_{K(i_2)}$
		\item 2) $\xi_{i_1, K(i_1)}\geq 0$ and $\xi_{i_2, K(i_2)}<0$, i.e., $0 \leq \theta_{i_1}<\omega_{K(i_1)}\leq \pi$ and $\pi+\theta_{i_2}<\omega_{K(i_2)}<2\pi$ and $0 \leq \theta_{i_2}<\pi$
		\item 3) $\xi_{i_1, K(i_1)}<0$ and $\xi_{i_2, K(i_2)}\geq 0$, i.e., $0 \leq \theta_{i_1}<\theta_{i_2}<\pi$ and $\pi<\omega_{K(i_1)}<\omega_{K(i_2)}\leq \theta_{i_2}+\pi$
		\item 4) $\xi_{i_1, K(i_1)}<0$ and $\xi_{i_2, K(i_2)}<0$, i.e., $0 \leq \theta_{i_1}<\theta_{i_2}<\pi$ and $\pi+\theta_{i_2} \leq \omega_{K(i_2)}<\omega_{K(i_1)}<2\pi$
	\end{itemize}
	
	If two lines, $\xi_{i_1, K(i_1)}$ and $\xi_{i_2, K(i_2)}$, cross, we will prove that the interchange of $K(i_1)$ and $K(i_2)$ will decrease the value of the right side of (\ref{circle1}). Because $[d^2_{\Gamma} (r_{i_1 l}, \mu_{K(i_1)})+d^2_{\Gamma} (r_{i_2 l}, \mu_{K(i_2)})]$ is proportional to $[\xi^2_{i_1, K(i_1)}+\xi^2_{i_2, K(i_2)}]$, we will use the latter instead of the former in the following discussion.\\
	%\begin{equation}
	%\label{local-change}
	%d^2_{\Gamma} (r_{i_1 l}, \mu_{K(i_1)}) +d^2_{\Gamma} (r_{i_2 l}, \mu_{K(i_2)}),
	%\end{equation}
	
	For case 1) and 4), we have $\xi^2_{i_1, K(i_1)}+\xi^2_{i_2, K(i_2)} = {(\omega_{K( i_1)}-\theta_{ i_1})}^2+{(\omega_{K( i_2 )}-\theta_{ i_2})}^2$. When $K(i_1)$ and $K(i_2)$ are switched, we have
		\begin{equation}
		\begin{aligned}
		& [{(\omega_{K(i_1)}-\theta_{i_1})}^2+{(\omega_{K(i_2)}-\theta_{i_2})}^2] \\
		&-[{(\omega_{K(i_2)}-\theta_{i_1})}^2+{(\omega_{K(i_1)}-\theta_{i_2})}^2]  \\
		&= 2(\theta_{i_2} - \theta_{i_1})(\omega_{K(i_1)}- \omega_{K(i_2)}) >0
		\end{aligned}
		\end{equation}
		Next for case 2) and 3), we have $\xi^2_{i_1, K(i_1)}+\xi^2_{i_2, K(i_2) }= {(\omega_{K( i_1)}-\theta_{ i_1}-2\pi)}^2+{(\omega_{K(i_2)}-\theta_{i_2})}^2$. Similarly, switching $K(i_1)$ and $K(i_2)$ results in
			\begin{equation}
			\begin{aligned}
			& [{(\omega_{K( i_1)}-\theta_{ i_1}-2\pi)}^2+{(\omega_{K(i_2)}-\theta_{i_2})}^2] \\
			&-[{(\omega_{K( i_2)}-\theta_{ i_1}-2\pi)}^2+{(\omega_{K(i_1)}-\theta_{i_2})}^2] \\
			&= 2(\theta_{i_2} - \theta_{i_1})(\omega_{K(i_1)}- \omega_{K(i_2)}) +4\pi(\omega_{K( i_2)}-\omega_{K( i_1)})>0
			\end{aligned}
			\end{equation}
			%where the second line comes from the fact that $\pi<\omega_{K(i_2)}\leq \theta_{i_2}+\pi<2\pi$.
			Thus, under four cases, we have proved that when two lines do not cross, the right side of (\ref{circle1}) gets a smaller value.
            \end{proof}

\begin{figure*}
\centering
\includegraphics[width=2 in,height=2 in]{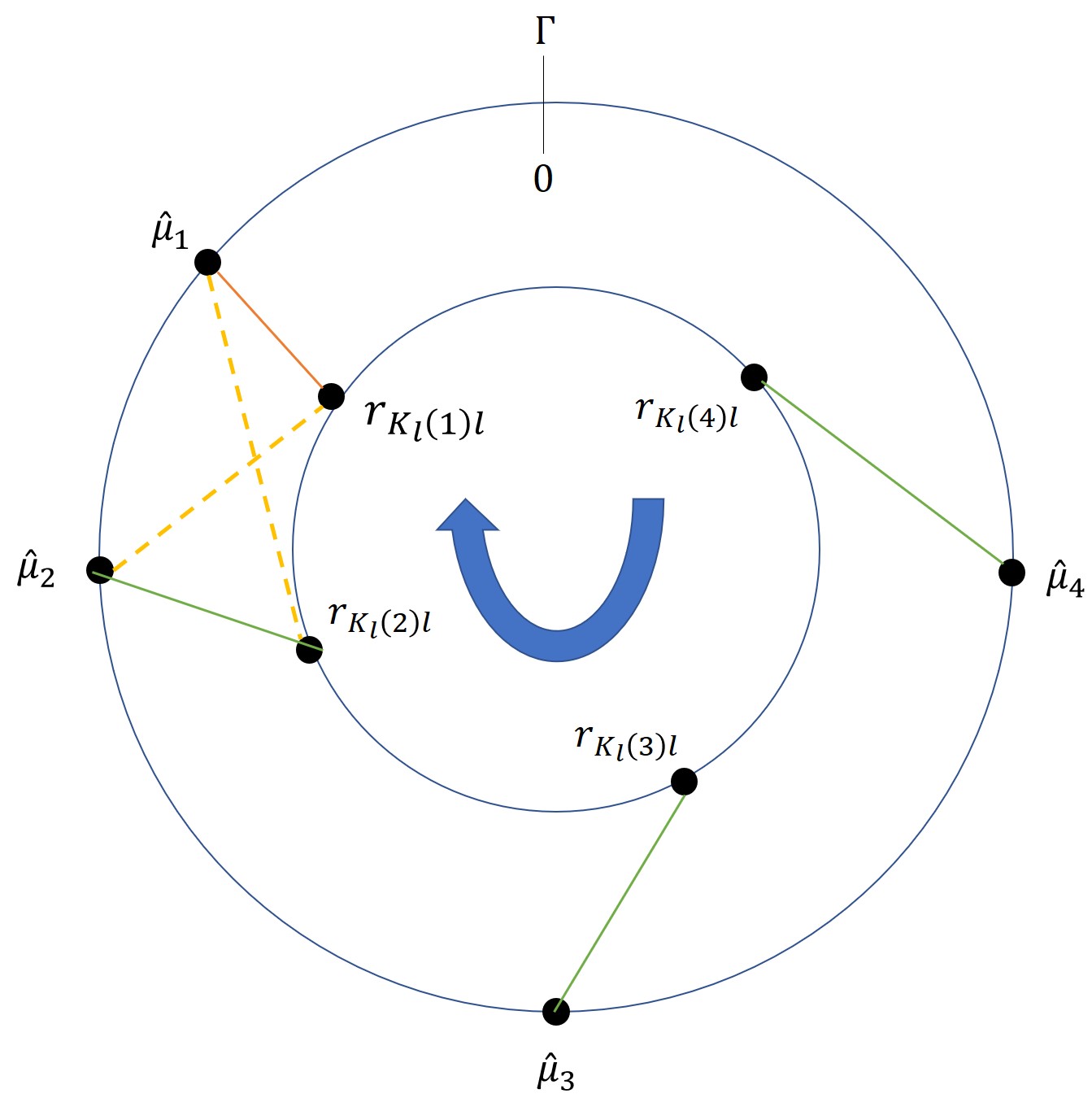}
\centering
\includegraphics[width=2 in,height=2 in]{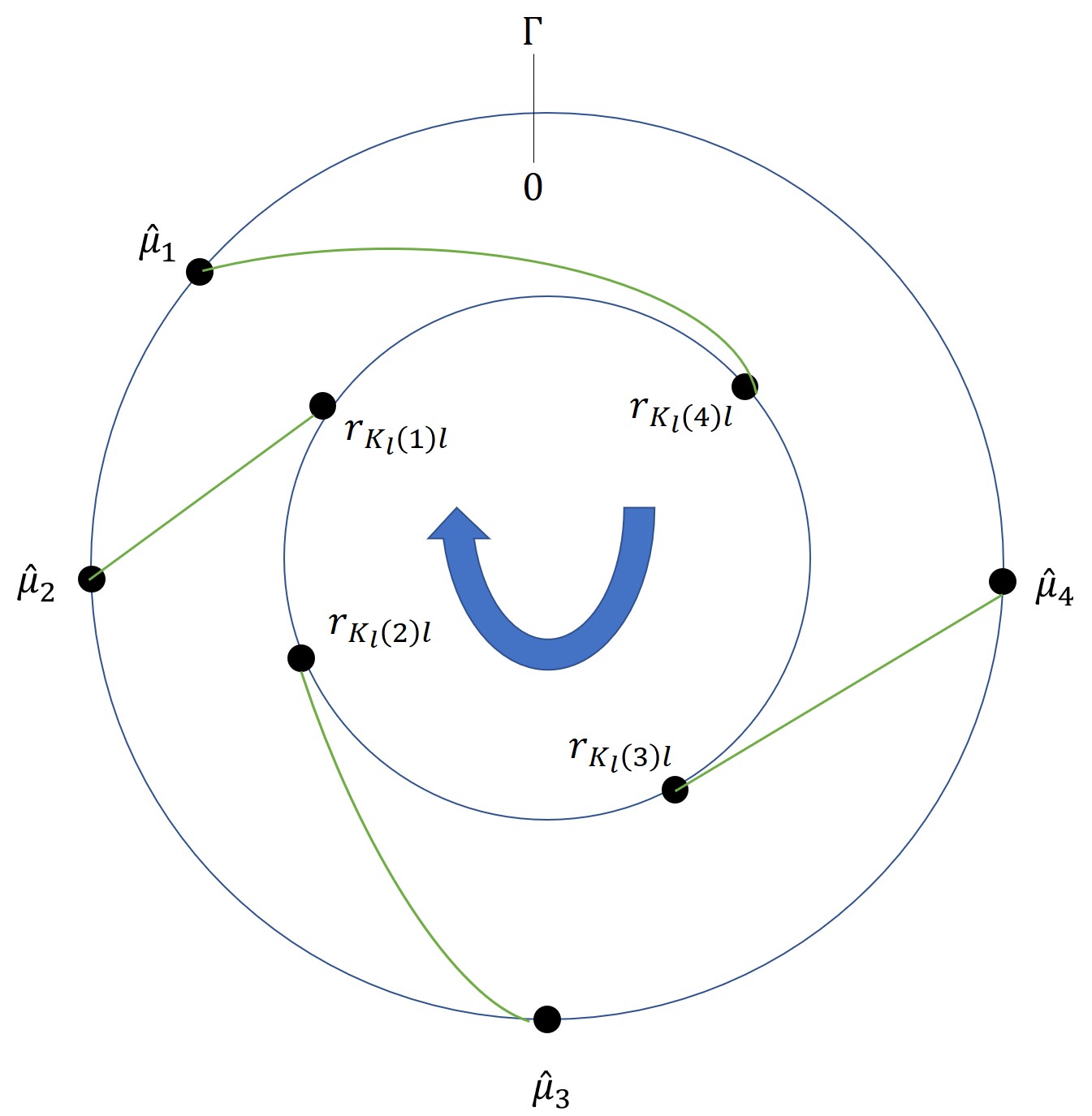}
\caption{Illustration for the Step one of Algorithm 2}
\label{concentric}
\end{figure*}

We now move to the second step: given residue clustering, how to figure out the optimal common residue? It is evident that with given residue clustering, the optimal issue is reduced to $N$ independent estimations for a single common residue. This problem has been previously studied in \cite{mle1}, where it proved that the optimal estimation can be determined in $O(L)$ complexity. For completeness, we present the skeleton of \cite{mle1} as follows with a simplified proof.

With given clustering $\bm{K}_{[1:L]}^{t+1}$, we need to figure out the optimal $\bm{{\mu}}^{t+1}_{[1:N]}$, i.e.,
\begin{equation}
\label{estmu}
\mu_i^{t+1} = \arg\min_{x \in [0,\Gamma)} \sum_{l=1}^{L} d^2_{\Gamma} (x, r_{k_l(i)l})
\end{equation}
where $x$ denotes a point on the 'small circle' modulo $\Gamma$. For simplicity, we assume that $\{\gamma_{(1)}, \gamma_{(2)}, ... ,\gamma_{(L)}\}$ denote $\{r_{k_l(i)l}\}$ in an ascending order. Based on the definition of $d_{\Gamma}$, there must exist $b_{l} \in \{0, \pm1\}$, such that
\begin{equation}
 \sum_{l=1}^{L} w_{l} d^2_{\Gamma} ({\mu^{t+1}_i} , \gamma_l) = \sum_{l=1}^{L} w_{(l)} (\mu^{t+1}_i - \gamma_l-b_{l}\Gamma)^2
\end{equation}
As $d_{\Gamma} ({\mu^{t+1}_i} , \gamma_{(l)}) \leq \frac{\Gamma}{2}$, we consider the interval $I_{\mu^{t+1}_i}=[{\mu^{t+1}_i}-\frac{\Gamma}{2}, {\mu^{t+1}_i}+\frac{\Gamma}{2})$. Without loss of generality, we assume $\mu^{t+1}_i \geq \frac{\Gamma}{2}$. Thus, there must some $(j)$ such that $\gamma_{(j)}, \gamma_{(j)+1}, ... , \gamma_{(L)}, \gamma_{(1)}+\Gamma, ... , \gamma_{(j-1)}+\Gamma$, which are also in an ascending order, all belong to $I_{\mu^{t+1}_i}$. In such case, $d_{\Gamma} ({\mu^{t+1}_i} , \gamma_{(l)}) = |{\mu^{t+1}_i}-\gamma_{(l)}|$ if $l \geq j$; otherwise $d_{\Gamma} ({\mu^{t+1}_i} , \gamma_{(l)}) = \gamma_{(l)}+\Gamma-{\mu^{t+1}_i}$. Substitute the above into (\ref{estmu}), then we have $\mu^{t+1}_i = \frac{\sum_{l=1}^L w_{(L)} \gamma_{(l)}}{\sum_{l=1}^L w_{(l)}} + \frac{\sum_{l=1}^j w_{(l)} \gamma_{(l)}}{\sum_{l=1}^L w_{(l)}}$. Since $j \in [1:L]$, therefore,
\begin{equation}
\label{M-step}
\mu^{t+1}_i \in \{\frac{\sum_{l=1}^L w_{(L)} \gamma_{(l)}}{\sum_{l=1}^L w_{(l)}} + \frac{\sum_{l=1}^j w_{(l)} \gamma_{(l)}}{\sum_{l=1}^L w_{(l)}}, j=0,1,...,L-1\},
\end{equation}
%and both ${\mu^{t+1}_i}$ and $\gamma_l$  within $[0,\Gamma)$, $d^2_{\Gamma} ({\mu^{t+1}_i} , \gamma_l)$ must fall in one of $\{({\mu^{t+1}_i} -\gamma_l)^2, ({\mu^{t+1}_i} - \gamma_l-\Gamma)^2 \}$.  In the following, we show that for the optimal case, there must exist some $l_0 \in \{1,2,...,L\}$ such that $l<l_0$, $b_l=1$ or $l \geq l_0$, $b_l =0$. Since the distance between ${\mu^{t+1}_i}$ and ($\gamma_l+b_{l}\Gamma$) is no bigger than $\frac{\Gamma}{2}$, the absolute difference between any two ($\gamma_l+b_{l}\Gamma$) should also be no bigger than $\Gamma$.First, it is not hard to observe that if $l_1 < l_2 \in \{1,2,...,L\}$, then $b_{l_1} \geq b_{l_2}$. Otherwise, if  $b_{l_2}=1$ and $b_{l_1}=0$, $|(\gamma_{l_2}+b_{l_2}\Gamma) - (\gamma_{l_1}+b_{l_1}\Gamma)| = |\gamma_{l_2}- \gamma_{l_1} + \Gamma| > \Gamma.$ Therefore, $\textbf{b} = ( b_1, b_2, ... ,b_L)$ must fall in one of the following candidates, $\textbf{v}_1 = (0,0, ... ,0,0), \textbf{v}_2 = (1,0, ... ,0,0),...,\textbf{v}_L=(1,1, ... ,1,0) \}$. On the other hand, when $\textbf{b} = \textbf{v}_j$, we have
which implies that in Step two, the complexity of estimating each $\mu^{t+1}_i$ is $O(L)$, and total complexity is $O(NL)$ . As a summary, given the estimations of $\bm{\mu}^t$, we can figure out the optimal clustering $\bm{K}^{t+1}_{[1:L]}$ with $O(NL)$ complexity according to Theorem 2. Relying on the estimations of $\bm{K}_{[1:L]}$, we can further determine the optimal $\bm{\mu}^{t+1}$ still in $O(NL)$ complexity. It noted that the number of candidates of $\bm{K}_{[1:L]}$ is finite and thus the algorithm will always converge to some stationary sate.  We conclude such iterative searching as follows.

\begin{algorithm}
\caption{MAP of Classification and Common Residues}
\textbf{Input:} Given moduli $m_l=\Gamma M_l$ and the residue observed $R_{il} $, $i=1,2,...,N$ and $l=1,2,...,L$.

1. Calculate ${r}_{il} = \langle  \widetilde{r}_{il} \rangle_{M_l}$.

2. Initialize $\{\hat{\mu}_i^0, i=1,2,...,N\}$.

%Start iteration $t$ starting from $1$ with an initialization for $\{\hat{\mu}_i(0), i=1,2,...,N\}$ by randomly drawing a set of observations, i.e., $\bm{r}_{[1:N],l} = \{r_{1,l}, r_{2,l},...,r{N,l}\}$ if observation set $l$ is drawn.

3. Begin iteration $t$ from $1$ to $T$:
    \begin{itemize}
     \item Step-1: Given $\{\hat{\mu}_i^{t-1}, i=1,2,...,N\}$, determine the optimal clustering, $\{K^{t}_l, l=1,2,...,L\}$, following Theorem \ref{circlemin1}.

     \item Step-2: Given $\{K^{t}_l, l=1,2,...,L\}$, update $\{\hat{\mu}_i^{t-1}, i=1,2,...,N\}$ by (\ref{estmu}).
     %\begin{equation}
%\hat{\mu}_i(t) = \arg \min_{x \in [0,\Gamma)} \sum_{l=1}^{L} d^2_{\Gamma}(x, \widetilde{r}_{K^{t}_l(i)l})
%\end{equation}
    %\item $t = t+1$.
    \end{itemize}

4. Calculate
\begin{equation}
    q_{il} = [ \frac{{R}_{K^{T}_l(i)l}-\widetilde{r}_{K^{T}_l(i)l}}{\Gamma} ]
\end{equation}
and reconstruct quotient $Q_i$ from $q_{il}$ with moduli $M_l$ via conventional CRT.

5. Reconstruct $ \hat{Y}_i = Q_i\Gamma + \hat{\mu}_i(T)$.

\textbf{Output:}  $\hat{Y}_i$, $i=1,2,...,N$.
\end{algorithm}

\begin{example}
Consider $m_1 = 5 \times 2$, $m_2 = 5 \times 3$, $m_3=5 \times 7$, $Y_1 =11$, $Y_2 = 18$ and $Y_3=64$. We set $w_1 = w_2 = w_3= 1$. Three observations are obtained $\bm{R}_1=\{2, 9, 4.3\}$, $\bm{R}_2=\{10, 3, 3.6\}$ and $\bm{R}_3=\{10.5, 19.1, 29.4\}$. Accordingly, one can derive $\bm{r}_1=\{2, 4, 4.3\}$, $\bm{r}_2=\{0, 3, 3.6\}$, $\bm{r}_3=\{0.5, 4.1, 4.4\}$. We initialize $\hat{\bm{\mu}}^{0}$ with $\bm{r}_1=\{2,4,4.3\}$, i.e., $\hat{\mu}^{0}_1 =2$, $\hat{\mu}^{0}_2 =4$ and $\hat{\mu}^{0}_3 =4.3$. $K^{1}_1$ is clear and for $K^{1}_2$, according to theorem 2, the optimal matching must be one of the following three cases in a rotation manner: (a) $(0 \to 2), (3 \to 4),(3.6 \to 4.3)$; (b) $(0 \to 4.3), (3 \to 2),(3.6 \to 4)$ and (c) $(0 \to 4), (3 \to 4.3),(3.6 \to 2)$. Clearly, (b) minimizes (\ref{circle1}). Similarly, we find that $(0.5 \to 2), (4.1 \to 4),(4.4 \to 4.3)$ is optimal for $K^1_3$. Given $\bm{K^{1}}_{[1:L]}$, we proceed to estimate $\bm{\mu^{1}}_{[1:N]}$. Here we only take $\mu^{1}_{1}$ as an example. With $\bm{K^{1}}_{[1:L]}$, $\{2,3,0.5\}$ are grouped together. $\mu^{1}_{1}$ must be one of the three $\{ \frac{0.5+2+3}{3}, \frac{0.5+2+3+5}{3}, \frac{0.5+2+3+5\times2}{3} \}$ and we find that $\mu^{1}_{1}=\frac{5.5}{3}$ minimizes (\ref{estmu}).
\end{example}

\section{Robustness Strengthening and Simulation}
\noindent In this section, we will introduce error correcting codes to further strengthen the robustness of proposed statistical RCRTs. As we stressed earlier, to find correct $\bm{K}_{[1:L]}$ plays the key role in reconstruction. Even if only one residue is not correctly clustered, it may compromise estimation performance heavily in CRT systems. %Thereby with respect to the classification mistakes, it will be better to be measured by Hamming weights. %On the other hand, it is noted that the two proposed schemes, beyond the last step, are searching for either MAP of classification or MAP of both classification and common residues $\mu_i$, and reconstruction is implemented finally with the assistance of the estimation.
A natural question is that when perfect residue clustering is not achievable, whether robust reconstruction is still possible. Fortunately, a previous work \cite{tvt-2019} has provided a positive answer to this question, as it will implement robust reconstruction under few residues with arbitrary errors. Assume that $L$ moduli are used, where $\{m_l=\Gamma M_l\}$ are in an ascending order, and $L_0$ is the smallest positive integer $L_0 \leq L$, such that $lcm(m_1,m_2,...,m_{L_0}) = \Gamma \prod_{l=1}^{L_0} M_l> D$. Said another way, if no errors exist in residues, $Y_i$ can be sufficiently recovered from the residues of any $L_0$ moduli. By taking a wrong classification with an arbitrary error happening to that residue, from \cite{tvt-2019}, we have the following theorem:

\begin{thm}[ \cite{tvt-2019} ]
\label{tvt}
Given $\bm{K}_{[1:L]} $, where at least ($L-\lfloor \frac{L-L_0}{2} \rfloor$) residues $\{R_{i[1:L]}\}$ of $Y_i$ are correctly clustered and $\max_l \Delta_{il} - \min_l  \Delta_{il} < \frac{\Gamma}{2}$ hold for those correctly clustered residues $\{R_{i[1:L]}\}$ for each $i$, then there exists a robust reconstruction scheme for $\hat{Y}_i$ such that $|\hat{Y}_i - Y_i| \leq \frac{3\Gamma}{4}$.
\end{thm}

\begin{remark}
It is worthy mentioning that with fixed $L_0$, increasing $L$, i.e., with more moduli (samplers), does not guarantee to continuously improve the performance of reconstruction since a larger $L$ always degrades the clustering accuracy. Moreover, besides the threshold ($L-\lfloor \frac{L-L_0}{2} \rfloor$) requirement of $\bm{K}_{[1:L]}$,  Theorem \ref{tvt} also assumes that $\max_l \Delta_{il} - \min_l  \Delta_{il} < \frac{\Gamma}{2}$ for correctly clustered residues, which fails in a higher probability with a larger $L$ as analyzed in Section II. %\footnote{A lower bound is shown in \cite{mle1} with respect to different weights. Here we just assume the weights are the same for simplicity.}
\end{remark}

%As $L$ increases, the probability that the span of all errors is smaller than $\frac{\Gamma}{2}$ can be referred to (4) by setting $\delta = \frac{\Gamma}{4}$.

In general, the reconstruction performance depends on factors including $N$, $L$, $L_0$, $\Gamma$, the noise and also the desirable computation power, of which the relationships are too complicated to be concisely expressed. However, if the noise is limited, adding redundant residues properly can always improve the performance. Besides, the other advantages of the error correction mechanism will be clear for the following majority voting based estimations.

In the rest of the section, we will show the mechanism to fully utilize the samples from multiple samplers. A natural idea is to regroup the moduli into different sets. Then, we use the residues from each set to implement Algorithm 1 or 2. Roughly speaking, the basic requirement is that each set should include at least $L_0$ moduli, of which the lcm is bigger than $D$ in order to achieve a valid reconstruction. Thus, we can obtain several estimated $\{\hat{Y}_i\}$ from different sets. If there are $\kappa$ such moduli sets, we can then correspondingly pick the $N$ most frequent numbers from all $\kappa N$ reconstructed numbers as the output. %Therefore, there exists a tradeoff between $\kappa$ and the reconstruction performance according to Theorem \ref{tvt}.
\begin{figure*}
\label{simu}

\centering
\subfigure[Proposed Statistical RCRT-1 and Deterministic RCRT \cite{TSP2018}]{

\begin{minipage}[t]{0.48\textwidth}
\centering
\centerline{\includegraphics[width=90mm]{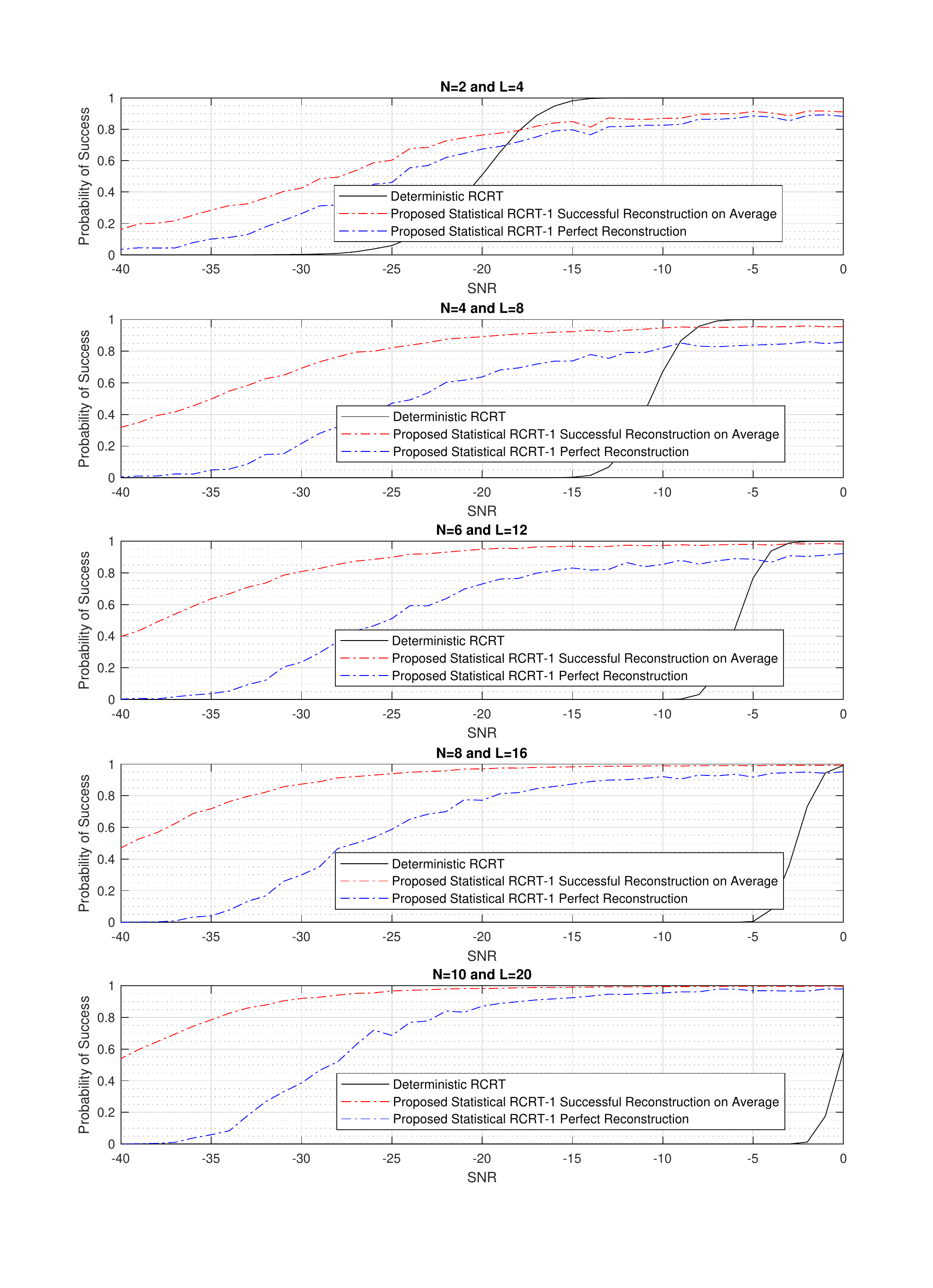}}
%\caption{Comparison between the proposed two statistical RCRT and Deterministic RCRT in \cite{TSP2018}}
\end{minipage}
}
\subfigure[Proposed Statistical RCRT-2 and Deterministic RCRT \cite{TSP2018}]{
\begin{minipage}[t]{0.48\textwidth}
\centering
\centerline{\includegraphics[width=90mm]{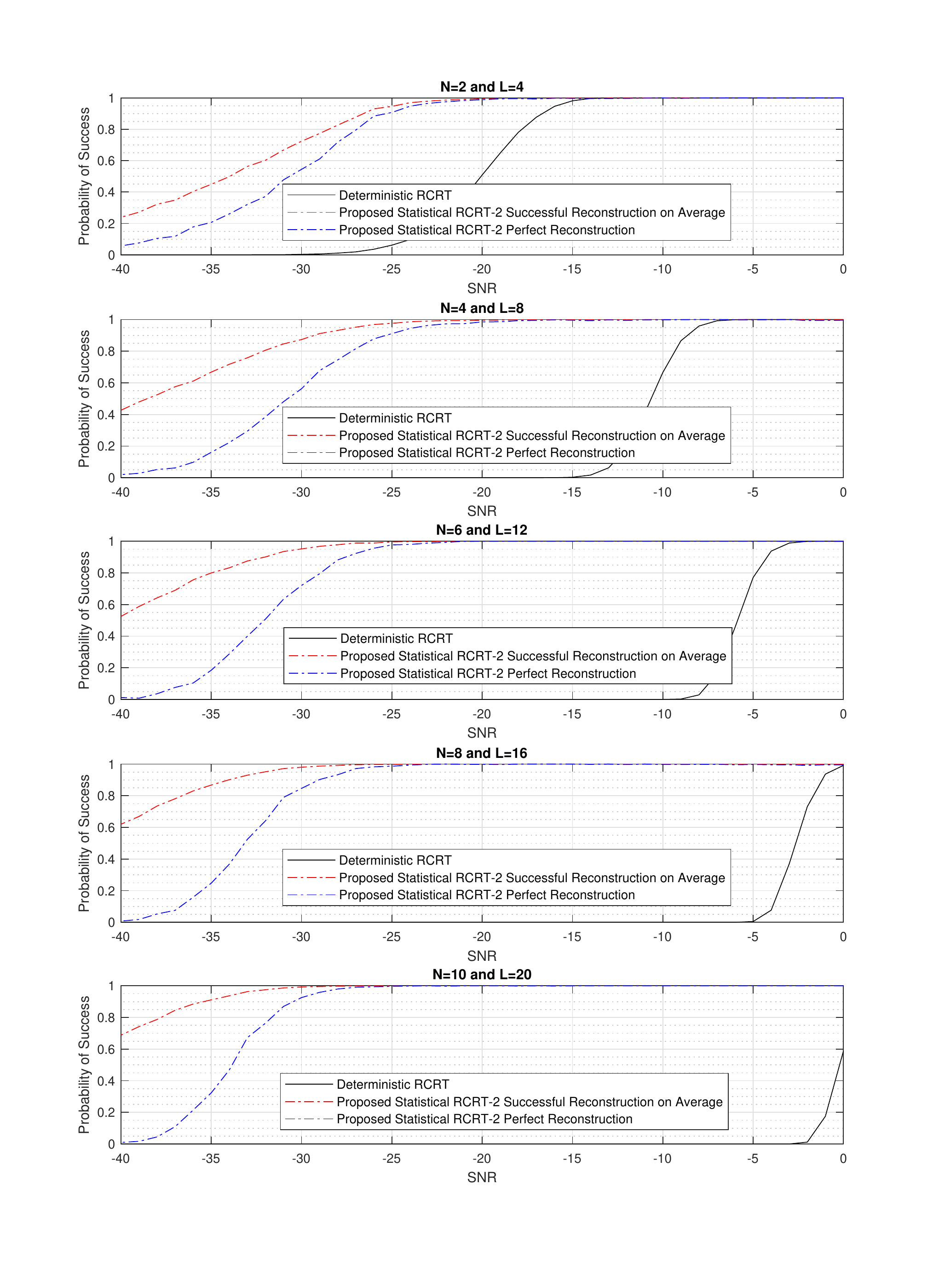}}
%\caption{Comparison between Proposed Statistical RCRT-2 and Deterministic RCRT}
\end{minipage}
}
\caption{Comparison for Success Rate of Robust Reconstruction between the proposed two statistical RCRT and Deterministic RCRT in \cite{TSP2018}}
\end{figure*}
However, even if both correct residue clusterings are obtained with two different sets, noise may still cause a slight difference between two estimated $\hat{Y}_i$ of the same $Y_i$, where the above idea can not be used straightforwardly. Instead, in our scheme, we focus on the $N$ most frequent quotients. \footnote{Even under correct residue classification, the common residues may be uniformly shifted by $\Gamma$ with the proposed scheme, depending on the choice of the cutting point $\tau$. Correspondingly, the difference of two reconstructions of the quotient $\lfloor \frac{Y_i}{\Gamma} \rfloor$ may be $1$. However, this can be easily distinguished and here we assume they share the same estimated quotient.}

%When the moduli are divided into $\kappa$ subsets without intersection, the estimations from different sets are independent. Let $B_j$ denote the event that $\lfloor \frac {Y_i}{\Gamma} \rfloor$ can be robustly reconstructed from the $j^{th}$ set and $\Pr(B_j)$ is a constant. Taking $B_j$ as a Bernoulli variable, $\sum_{j=1}^{\kappa} B_j > \frac{\Gamma}{2}$ is sufficient to show that the correct reconstruction of  $\lfloor \frac{Y_i}{\Gamma} \rfloor$ is one of the estimations with the highest $N$ frequencies. From Chernoff bound,

%\begin{lm}[Chernoff Bound for Bernoulli Variables]
%If $E[B_j]=p>\frac{1}{2}$ for $j=1,2,...,\kappa$, then
%\begin{equation}
%\label{chernoff}
%\Pr( \sum_{j=1}^{n} B_j > \frac{\kappa}{2}) \geq 1-e^{-\frac{\kappa}{2p} (p-\frac{1}{2})^2},
%\end{equation}
%\end{lm}

In the simulation, we compare proposed algorithms with the deterministic RCRT for multiple numbers in \cite{TSP2018}. %From a bird's eye on the deterministic RCRT, the restriction on the error magnitude basically is a sufficient condition which makes Condition 1 hold. If $|\max_{il} \Delta_{il}| < \frac{\Gamma}{4N}$, then any interval between two consecutive ${r}_{il}$ on the circle modulo $\Gamma$ no less than $\frac{\Gamma}{2N}$ must not be overlapped by any $I_i$. The results from  \cite{TSP-2018} show that once a non-overlapped point over the circle is found, with enough redundancy, the $X_i$ can be reconstructed with a reconstruction error upper bounded by $ \frac{\Gamma}{4N}$. Obviously, when $|\max_{il} \Delta_{il}| < \frac{\Gamma}{4N}$ holds, the interval between two consecutive ${r}^c_{il}$ with the longest length must be non-overlapped. Therefore, we modify the RCRT in \cite{TSP-2018} slightly by regarding the one of $r_{il}$ which is with the longest distance to its neighbors as the a cutting point��.
The following simulation results show the comparison among Statistical RCRT-1, short for Algorithm 1 (MAP of residue clustering), Statistical RCRT-2, short for Algorithm 2 (MAP of both residue clustering and common residue estimation), and deterministic one \cite{TSP2018}. Here, we set $L_0=2$. Moduli are selected in a form where $\Gamma=100$ and $\{M_1, M_2, ... ,M_L\}$ are the sequence of primes starting from 23. Accordingly, $\bm{Y}_{[1:L]}$ are randomly selected from $[0,66700]$, where $D=100 \times 23 \times 29$. Referring to the requirements of moduli in \cite{TSP2018}, the lcm of all moduli is larger than the product of $\{Y_i\}$ and we set $L=L_0N=2N$. We assume the variance of noise $\sigma^2_l=\sigma^2 = 10^{-SNR/10}$ for each $l$. As for the simulation shown in Fig. 5, for both Algorithm 1 \& 2, we utilize residues from each pair of moduli, in total $\binom{2N}{2}$ many groups for the simulations.

We define that $Y_i$ is {\em robustly reconstructed}  if the reconstruction error is upper bounded by $\Gamma$. Since the reconstruction of $N$ numbers is finally converted to $N$ independent reconstruction processes for each $Y_i$ in both proposed statistical RCRTs, we define the {\em success rate on average} as the expected rate that a $Y_i$ can be robustly recovered. Similarly, the {\em perfect reconstruction rate} is the probability that all $\{Y_i\}$ are robustly recovered. The two metrics may be of different interest in different applications. SNR is limited within $[-40, 0]$ and $N \in \{2, 4, 6, 8,10\}$. We run 1000 simulations to estimate the success rate in each scenario.

The simulation results for both Algorithm 1 and Algorithm 2 compared to deterministic RCRT in \cite{TSP2018} are presented in Fig. 5 (a) and (b), respectively. From Fig. 5, Algorithm 2 outperforms Algorithm 1 as analyzed before, since Algorithm 2 is in an iterative manner, which may face a little more computational overhead. Heuristically, due to the non-weighted nature of CRT,  if errors happen to both residue clusterings in two different moduli sets, the resultant $\hat{Y}_i$ and quotients $\lfloor \frac{\hat{Y}_i}{\Gamma} \rfloor$ associated will be dramatically different in very high probability. Therefore, when estimated $\lfloor \frac{\hat{Y}_i}{\Gamma} \rfloor$ has been reconstructed at least twice across the sets, it is of high confidence to be selected in the majority voting. Our simulations coincide with such intuition.

The distributions of the iteration number for different $N$ in Algorithm 2 are shown in the top subfigure of Fig. \ref{MLEsim}. Generally, Algorithm 2 can reach a stationary state within $10$ iterations, mostly concentrated between $2$ and $3$. To be more detailed, the other successive five subfigures in Fig. \ref{MLEsim} show how the noise level influences the number of iterations. Here 'low SNR' stands for the cases where SNR is within $[-40,-20)$ and 'high SNR' refers to SNR within $[-20,0]$. Clearly, low SNR may incur a higher complexity, while the number of iterations in high SNR cases is within 3 rounds on average.

Finally, we give two examples to show how error correcting techniques can improve the performance, where $N$ is set to $2$ and $6$, respectively. Given $L=4$, we compare the performances of Algorithm 2 with and without incorporating error correction \footnote{From Theorem \ref{tvt}, here we can tolerate at most one clustering error since $L_0=2$ and $\lfloor \frac{L-L_0}{2} \rfloor =1$.}, which is shown in Fig. \ref{error-correction}. It is noted that all the four moduli need to be simultaneously used in Algorithm 2 to apply error correction in reconstruction; while without error correction, the moduli are regrouped into $\binom{4}{2}$, i.e., 6, sets and the original Algorithm 2 is applied on each set with majority voting strategy afterwards. In the case with $N=2$ (the upper one in Fig. \ref{error-correction}), error correction does not provide a better tradeoff. That is because when $N$ is small, clusterings produced are accurate enough. In such scenario, it is more reasonable to generate more modulus sets for the reconstructions of numbers respectively to the majority voting. However, as $N$ increases, clustering errors happen in a sharply increasing rate, where merely a larger number of reconstructions from different sets do not benefit the success rate of majority voting so much. The lower one in Fig. \ref{error-correction} shows that given $N=6$, the error correction based Statistical RCRT-2 outperforms the original one when SNR is bigger than -37.5.

%\begin{figure*}
%\label{map1}
%	\vskip 0.2in
%	\begin{center}
	%	\centerline{\includegraphics[width=90mm]{MAP.eps}}
%		\caption{Comparison between Proposed Statistical RCRT-1 and Deterministic RCRT}
%	\end{center}
%	\vskip -0.2in

%	\begin{center}
	%	\centerline{\includegraphics[width=90mm]{MLE.eps}}
	%	\caption{Comparison between Proposed Statistical RCRT-1, RCRT-2 and Deterministic RCRT}
%	\end{center}
%	\vskip -0.2in
%\end{figure*}

\begin{figure}
\label{iteration}
	\vskip 0.2in
	\begin{center}
		\centerline{\includegraphics[width=90mm]{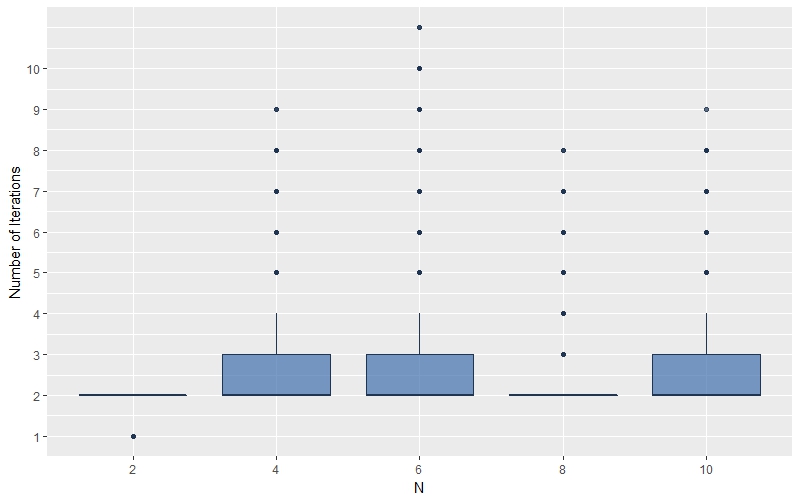}}
		\centerline{\includegraphics[width=90mm]{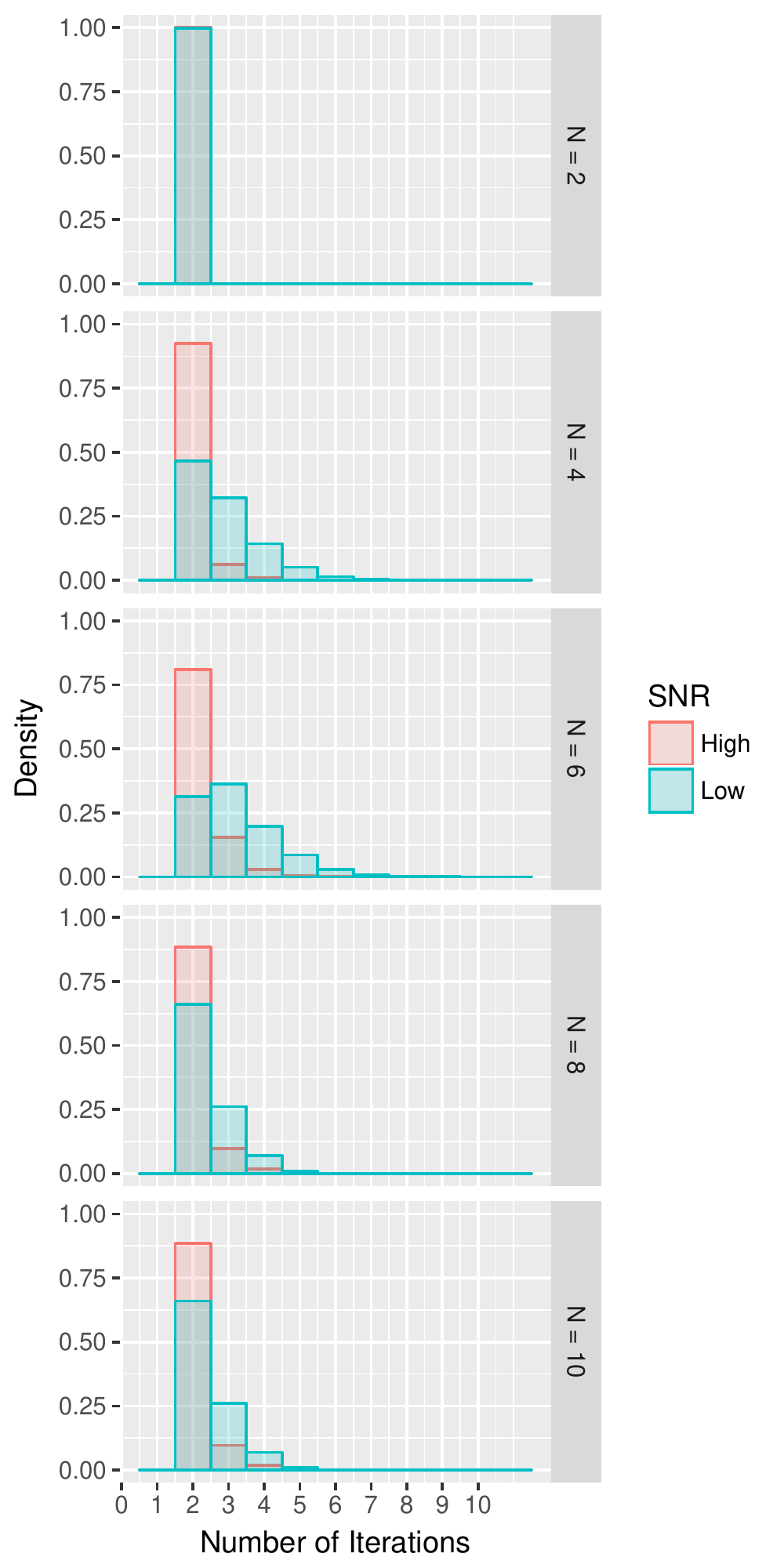}}
		\caption{The Number of Average Iterations of Algorithm 2}
		\label{MLEsim}
	\end{center}
	\vskip -0.2in
\end{figure}

\begin{figure}
\label{thres}
	\vskip 0.2in
	\begin{center}
		\centerline{\includegraphics[width=90mm]{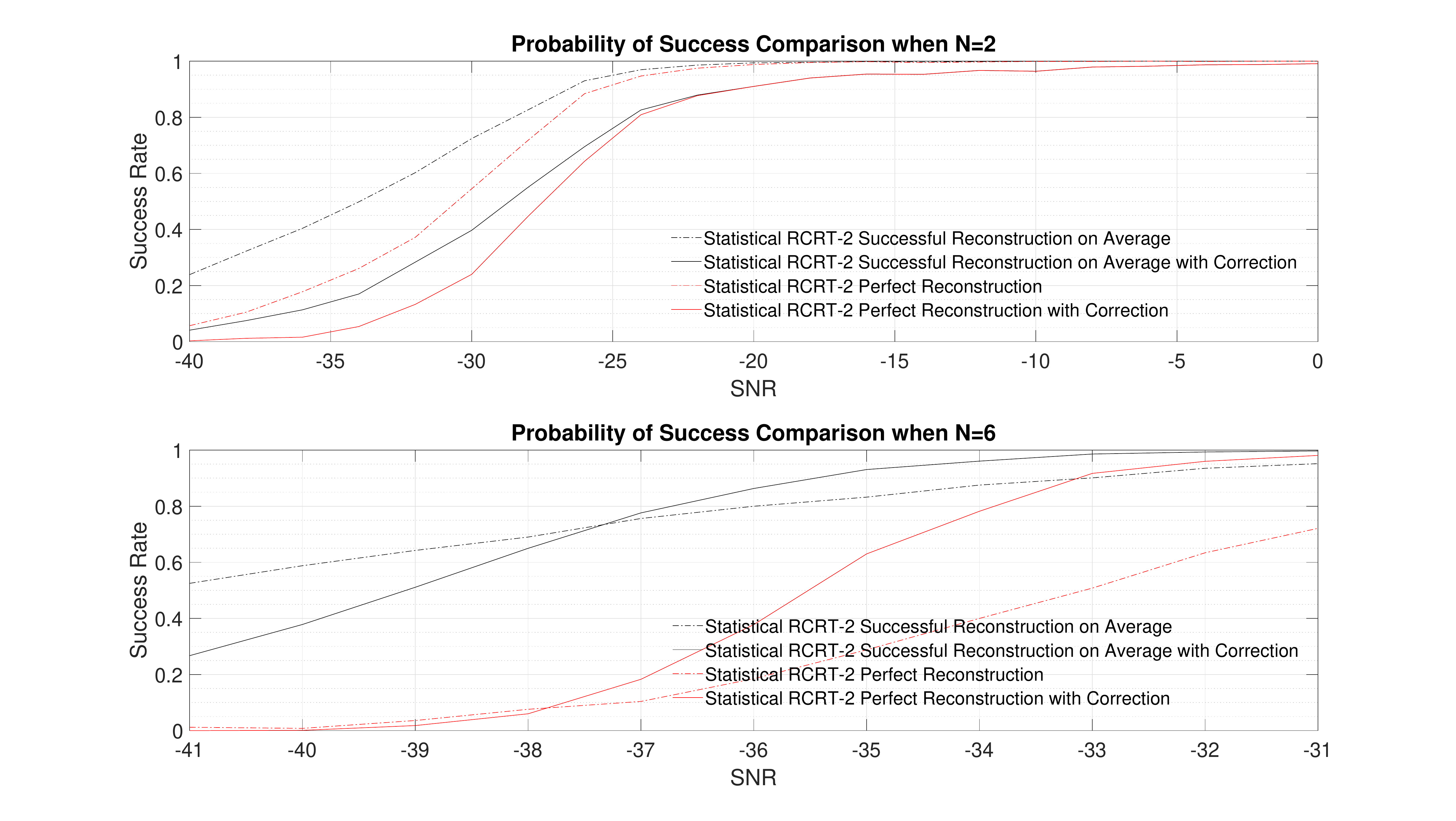}}
		\caption{ Proposed Statistical RCRT-2 with Error Correction }
		\label{error-correction}
	\end{center}
	\vskip -0.2in
\end{figure}

\section{Conclusion and Prospects}
\noindent In this paper, we present the first statistical based approaches to efficiently solve the robust reconstruction of multiple numbers from unordered residues. Compared with deterministic schemes, the proposed statistical RCRT methods significantly improve the performance of reconstruction, which can be further strengthened with error correcting techniques. However, in extremely low noise cases, the performance of deterministic schemes can be better especially when $\Gamma$ is small. Therefore, it would be of great interest to investigate the tradeoff between statistical inference and deterministic error tolerance.

Another problem that remains open is how to determine the optimal size of modulus set for reconstruction. We believe it is nontrivial to describe the tradeoff between the clustering accuracy compromise and the additional robustness gained from redundant moduli for error correction.

\bibliographystyle{plain}
\bibliography{ref}

\end{document}